\documentclass[conference]{IEEEtran}

\usepackage{amsmath,amssymb,amsfonts,amsthm}
\usepackage{booktabs}
\usepackage{cite}
\usepackage{graphicx}
\usepackage{url}
\usepackage{tabularx}
\usepackage{mathtools}
\usepackage{hyperref}
\usepackage{xcolor}
\usepackage{array}
\usepackage{multirow}
\usepackage{algorithm}
\usepackage{comment}
\usepackage[noend]{algpseudocode}
\usepackage{tikz}
\usetikzlibrary{arrows.meta,positioning,calc,decorations.pathreplacing, shapes.geometric,shapes.misc}
\usepackage{pgfplots}
\pgfplotsset{compat=1.17}

\newtheorem{definition}{Definition}
\newtheorem{lemma}{Lemma}
\newtheorem{proposition}{Proposition}
\newtheorem{theorem}{Theorem}
\newcommand{\Adv}{\mathsf{Adv}}
\newcommand{\Enc}{\mathsf{Enc}}
\newcommand{\Dec}{\mathsf{Dec}}
\newcommand{\Send}{\mathsf{Send}}
\newcommand{\Recv}{\mathsf{Recv}}
\newcommand{\Init}{\mathsf{Init}}
\newcommand{\DATA}{\mathsf{DATA}}
\newcommand{\DUMMY}{\mathsf{DUMMY}}
\newcommand{\FIN}{\mathsf{FIN}}
\newcommand{\ACK}{\mathsf{ACK}}
\newcommand{\Rand}{\mathsf{Rand}}

\newcommand{\sys}{\textsc{BiFEP}}
\newcommand{\dgsys}{\textsc{Bi-DG-FEP}}
\newcommand{\nActiveStreamTrials}{360}

\newcommand{\nDgAssertions}{3,328}
\newcommand{\nDgCases}{1,532}

\newcommand{\nDgPerfGood}{420}
\newcommand{\nDgPerfLatUs}{1.36}

\newcommand{\nLoadExpMaxUtil}{1.09}

\newcommand{\nLossAckFive}{$\bot$\,/\,4}
\newcommand{\nLossAckFour}{$\bot$\,/\,4}
\newcommand{\nLossAckOne}{4\,/\,4}
\newcommand{\nLossAckSix}{$\bot$\,/\,4}
\newcommand{\nLossAckThree}{$\bot$\,/\,4}
\newcommand{\nLossAckTwo}{$\bot$\,/\,4}
\newcommand{\nLossLingerAckFive}{12\,/\,8}
\newcommand{\nLossLingerAckFour}{12\,/\,8}
\newcommand{\nLossLingerAckOne}{8\,/\,8}
\newcommand{\nLossLingerAckSix}{$\bot$\,/\,8}
\newcommand{\nLossLingerAckThree}{12\,/\,8}
\newcommand{\nLossLingerAckTwo}{12\,/\,8}
\newcommand{\nLossLingerOneFive}{12\,/\,12}
\newcommand{\nLossLingerOneFour}{12\,/\,12}
\newcommand{\nLossLingerOneOne}{8\,/\,8}
\newcommand{\nLossLingerOneSix}{16\,/\,16}
\newcommand{\nLossLingerOneThree}{12\,/\,12}
\newcommand{\nLossLingerOneTwo}{12\,/\,12}
\newcommand{\nLossLingerSymFive}{12\,/\,12}
\newcommand{\nLossLingerSymFour}{12\,/\,12}
\newcommand{\nLossLingerSymOne}{8\,/\,8}
\newcommand{\nLossLingerSymSix}{16\,/\,16}
\newcommand{\nLossLingerSymThree}{12\,/\,12}
\newcommand{\nLossLingerSymTwo}{12\,/\,12}
\newcommand{\nLossLingerSymZero}{8\,/\,8}
\newcommand{\nLossLingerThreeAckFive}{20\,/\,16}
\newcommand{\nLossLingerThreeAckFour}{20\,/\,16}
\newcommand{\nLossLingerThreeAckOne}{16\,/\,16}
\newcommand{\nLossLingerThreeAckSix}{24\,/\,16}
\newcommand{\nLossLingerThreeAckThree}{20\,/\,16}
\newcommand{\nLossLingerThreeAckTwo}{20\,/\,16}
\newcommand{\nLossLingerThreeOneFive}{20\,/\,20}
\newcommand{\nLossLingerThreeOneFour}{20\,/\,20}
\newcommand{\nLossLingerThreeOneOne}{16\,/\,16}
\newcommand{\nLossLingerThreeOneSix}{24\,/\,24}
\newcommand{\nLossLingerThreeOneThree}{20\,/\,20}
\newcommand{\nLossLingerThreeOneTwo}{20\,/\,20}
\newcommand{\nLossLingerThreeSymFive}{20\,/\,20}
\newcommand{\nLossLingerThreeSymFour}{20\,/\,20}
\newcommand{\nLossLingerThreeSymOne}{16\,/\,16}
\newcommand{\nLossLingerThreeSymSix}{24\,/\,24}
\newcommand{\nLossLingerThreeSymThree}{20\,/\,20}
\newcommand{\nLossLingerThreeSymTwo}{20\,/\,20}
\newcommand{\nLossLingerThreeSymZero}{16\,/\,16}
\newcommand{\nLossLingerTwoAckFive}{16\,/\,12}
\newcommand{\nLossLingerTwoAckFour}{16\,/\,12}
\newcommand{\nLossLingerTwoAckOne}{12\,/\,12}
\newcommand{\nLossLingerTwoAckSix}{20\,/\,12}
\newcommand{\nLossLingerTwoAckThree}{16\,/\,12}
\newcommand{\nLossLingerTwoAckTwo}{16\,/\,12}
\newcommand{\nLossLingerTwoOneFive}{16\,/\,16}
\newcommand{\nLossLingerTwoOneFour}{16\,/\,16}
\newcommand{\nLossLingerTwoOneOne}{12\,/\,12}
\newcommand{\nLossLingerTwoOneSix}{20\,/\,20}
\newcommand{\nLossLingerTwoOneThree}{16\,/\,16}
\newcommand{\nLossLingerTwoOneTwo}{16\,/\,16}
\newcommand{\nLossLingerTwoSymFive}{16\,/\,16}
\newcommand{\nLossLingerTwoSymFour}{16\,/\,16}
\newcommand{\nLossLingerTwoSymOne}{12\,/\,12}
\newcommand{\nLossLingerTwoSymSix}{20\,/\,20}
\newcommand{\nLossLingerTwoSymThree}{16\,/\,16}
\newcommand{\nLossLingerTwoSymTwo}{16\,/\,16}
\newcommand{\nLossLingerTwoSymZero}{12\,/\,12}
\newcommand{\nLossOneFive}{8\,/\,8}
\newcommand{\nLossOneFour}{8\,/\,8}
\newcommand{\nLossOneOne}{4\,/\,4}
\newcommand{\nLossOneSix}{12\,/\,12}
\newcommand{\nLossOneThree}{8\,/\,8}
\newcommand{\nLossOneTwo}{8\,/\,8}
\newcommand{\nLossSymFive}{8\,/\,8}
\newcommand{\nLossSymFour}{8\,/\,8}
\newcommand{\nLossSymOne}{4\,/\,4}
\newcommand{\nLossSymSix}{12\,/\,12}
\newcommand{\nLossSymThree}{8\,/\,8}
\newcommand{\nLossSymTwo}{8\,/\,8}
\newcommand{\nLossSymZero}{4\,/\,4}

\newcommand{\nPerfBdEpochs}{7,680}
\newcommand{\nPerfBdGood}{28.2}
\newcommand{\nPerfBdInputPerEpoch}{1,100}
\newcommand{\nPerfBdLat}{37.0}
\newcommand{\nPerfBdOutputPerEpoch}{2,200}
\newcommand{\nPerfBdThpt}{56.4}

\newcommand{\nSchedLatMax}{234}
\newcommand{\nSchedLatMin}{20}

\newcommand{\nSourceHashBreak}{\texttt{94f9c56fb94d8911f7a5fc4f45f4a857}\allowbreak\texttt{ff5f78ff8669d31e87c5c88b8739b84f}}

\newcommand{\nStreamAssertionsTotal}{62,119}

\newcommand{\nStreamTrials}{491}
\newcommand{\nSweepAssertions}{442}
\newcommand{\nSweepPoints}{93}

\newcommand{\nTimedJitterRange}{4.2--6.7}

\newcommand{\nTimedMissRange}{77--89}

\newcommand{\nTimedOverFive}{56}
\newcommand{\nTimedOverTen}{21}

\newcommand{\nTotalAssertions}{65,889}

\title{Closing the Loop: Bidirectional Fully Encrypted Protocols}

\author{\IEEEauthorblockN{Baigang Chen}
	\IEEEauthorblockA{University of Minnesota\\
		chen9464@umn.edu}
	\and
	\IEEEauthorblockN{Nicholas Hopper}
	\IEEEauthorblockA{University of Minnesota\\
		hoppernj@umn.edu}}
        
\IEEEoverridecommandlockouts
\makeatletter\def\@IEEEpubidpullup{6.5\baselineskip}\makeatother
\IEEEpubid{\parbox{\columnwidth}{
Network and Distributed System Security (NDSS) Symposium 2027\\
22--26 March 2027, Seoul, Republic of Korea\\
ISBN 978-1-970672-09-1\\
https://dx.doi.org/10.14722/ndss.2027.240467\\
www.ndss-symposium.org
}
\hspace{\columnsep}\makebox[\columnwidth]{}}

\begin{document}
\maketitle

\begin{abstract}
Fully encrypted protocols (FEPs) provide encrypted channels that make all protocol-generated bytes computationally indistinguishable from uniform random strings. Several previous works have explored security definitions and constructions of unidirectional FEPs: protocols in which one party acts only as a sender, and the other acts only as a receiver.  However, most applications require two-way information exchange, and a network adversary can observe communication in both directions and their shared lifetime.  Because the semantics of bidirectional channels involve more complex shared state, it is possible that the ``na\"ive'' composition of two unidirectional channels can result in a two-way protocol that can be detected based on dependencies between the two directions, such as traffic imbalance, channel closure, failures, or connection tear-down.

To address this issue, we introduce new formal security definitions for bidirectional FEPs that capture exact shaping, delivery, protocol-state integrity, private half-close, and cross-direction isolation, while revealing a public ``sending schedule'' and ``closing epoch'' that may be randomized.  We show that the trivial composition fails to meet these definitions, leading to practical detection attacks.  We then construct provably secure bidirectional FEPs (BiFEPs) for both the datastream and datagram settings. For datastream, we combine two direction-separated FEPs with a ``wrapper'' layer that prevents detection based on the mismatch between uni- and bi-directional connection states.  For datagram, we add encrypted DATA/FIN/ACK with replay protection and loss-tolerant close. We validate the design through a Rust implementation and show that none of the surveyed deployed protocols provides the full set of BiFEP security properties.

\end{abstract}

\IEEEpeerreviewmaketitle

\section{Introduction}
\label{sec:intro}

Encryption hides message contents, but not necessarily the protocol carrying them. For instance, TLS and QUIC retain recognizable handshakes and other structured features that contribute to a protocol’s observable wire fingerprint~\cite{rfc8446,rfc9000,rfc9312}. An observer can use this fingerprint to classify, throttle, or block traffic without decrypting it. Obfuscated transports such as obfs4 and Shadowsocks seek to remove fixed protocol markers by making all protocol-controlled bytes appear random~\cite{angel2014obfs4,shadowsocks2022sip022}. Random-looking bytes alone, however, do not make a connection unidentifiable: message lengths, timing, and responses to active probes can remain distinctive, and some censors specifically target high-entropy traffic~\cite{wu2023great,wang2015seeing}. Protocol mimicry is likewise fragile when its imitation differs detectably from the target protocol~\cite{houmansadr2013parrot}. These limitations motivate a precise separation between cryptographically hiding protocol-generated bytes and controlling the metadata that remains visible.

Fenske and Johnson formalized this goal with \emph{fully encrypted protocols}
(FEPs), which make protocol-generated bytes indistinguishable from uniform
strings of the same public length~\cite{fenske2024bytes}. Their definitions,
however, model only one direction at a time. Two secure FEPs run in parallel
can therefore leak session structure through their interaction. In a
request--response exchange, for example, a na\"ive composition may cause the request direction to go silent when the requester half-closes while the response direction remains active, exposing that private event although every transmitted byte looks random. Cross-direction
reactions can similarly enable active probes. We formalize and measure this
gap (\S\ref{sec:defs}, \S\ref{sec:eval-close-leak}), motivating a single
bidirectional security object that jointly governs shaping, failure, and close
behavior.

We construct {\em bidirectional} FEPs (BiFEPs) for both datastreams and datagrams. The core difficulty in recognizing this extension is allowing the connection to close without revealing the private events that made it ready to close. If an authenticated FIN or a decryption failure immediately closes the transport, an observer learns when the private protocol state changed. An active adversary can also use this response as a probe. A protocol that never closes avoids this leakage but is not practical. Our datastream construction separates private close readiness from the public close time. Each endpoint continues the public traffic schedule after a half-close, and the transport closes only at the next epoch in a public schedule. The observer learns a quantized close time, but not the exact time of the private close event. This design treats termination as part of the secure channel~\cite{boyd2017secure}.

Datagrams require a separate construction because UDP may lose, reorder, or duplicate packets \cite{postel1980rfc0768}. A datastream FIN is drained once its ciphertext prefix leaves the sender, but a transmitted datagram may never arrive. Treating transmission as delivery could therefore cause the two endpoints to reach inconsistent close states. Our datagram construction places encrypted FIN and ACK bits inside authenticated, scheduled datagrams. It retransmits these control bits idempotently and requires authenticated evidence from both directions before closing.
After becoming ready to close, an endpoint remains live through a public number of additional close buckets and retransmits its acknowledgment. A later ACK-bearing datagram can therefore repair an isolated loss. Packet loss may prevent termination, but it cannot create false close readiness.

\textbf{Contributions.}
\begin{itemize}
\item \textbf{Security framework and proofs.}
We formalize new security conditions for the bidirectional case, including correctness, traffic shaping, passive and active security, protocol-state integrity, and private close.   We show that trivial composition of two uni-directional FEPs does not satisfy these security goals.

\item \textbf{Bidirectional constructions.}
We construct BiFEPs for both datastreams and datagrams, providing exact per-direction traffic schedules and encrypted close coordination tailored to each transport model.  We prove both constructions secure while exposing only the public schedule and quantized close epochs.

\item \textbf{Implementation and evaluation.}
We implement both constructions in Rust and evaluate their behavior under active modification, replay, loss, and timing variation. We also compare TLS~1.3~\cite{rfc8446}, QUIC~v1~\cite{rfc9000,rfc9312,thomson2021rfc}, WireGuard~\cite{donenfeld2017wireguard}, Shadowsocks~2022~\cite{shadowsocks2022sip022}, obfs4~\cite{angel2014obfs4}, and Tor~\cite{torproject2026specs} at the specification level. None provides all bidirectional FEP properties. We further quantify the costs and practical constraints of cover traffic.
\end{itemize}

\section{Background}
\label{sec:background}

\textbf{FEP security notions.}
Fenske and Johnson formalize FEP security for both datastreams and datagrams~\cite{fenske2024bytes}. Passive security replaces each sender output with a uniform random string of the same length and requires indistinguishability of the transcripts. The active datastream experiment also exposes a receiver oracle and tracks whether the adversary-delivered stream remains a prefix of the honestly shown stream. Oracle outputs are suppressed while synchronized; after the first deviation, the ideal experiment ceases decryption, so any real-world semantic output is a forgery.
Datastream \emph{traffic shaping} requires exactly the requested number of bytes per sender call, making length and timing explicit inputs. The corresponding datagram notion (FEP-CCA) covers atomic packets with chosen-ciphertext access. Channel-security work under fragmentation~\cite{fischlin2015data,boldyreva2012security} and for bidirectional channels~\cite{marson2017security} studies related correctness and integrity questions without additionally requiring a random wire image.

\textbf{The datastream FEP.}
Our bidirectional datastream construction uses the datastream FEP of~\cite{fenske2024bytes} as its inner layer. One abstract record with body $u$ consumes two AEAD sequence numbers: with record counter $r$ and key $k$,
\begin{equation}
\begin{aligned}
 H_r&=\Enc_k\!\left(2r,\ \mathsf{BE}_{16}(|C_r|)\right),\\
 C_r&=\Enc_k\!\left(2r+1,\ \mathsf{BE}_{16}(\ell_p)\,\Vert\,0^{\ell_p}\,\Vert\,u\right),
\end{aligned}
 \label{eq:inner}
\end{equation}
where $\mathsf{BE}_{w}$ denotes a $w$-bit big-endian encoding, so $\mathsf{BE}_{16}$ occupies two bytes, and $\ell_p$ is an inner-layer padding length. The sender buffers generated ciphertext and releases a prefix of the requested length on each invocation. In \S\ref{sec:stream}, $\mathsf{UD.Enc}_k(r,u)$ denotes this sender's complete record-encoding step; the bidirectional wrapper owns the ciphertext queue and scheduled prefix release. Because the wrapper supplies authenticated cover, we set $\ell_p\equiv 0$ and reject nonzero padding lengths. With AES-256-GCM, the inner per-record overhead is $c_{\mathrm{enc}}=36$ bytes. Under a length-additive IND\$-CPA and INT-CTXT AEAD, the layer provides datastream shaping and active FEP security~\cite{fenske2024bytes}.

\textbf{The atomic datagram FEP.}
The datagram FEP of~\cite{fenske2024bytes} carries a fresh transmitted nonce in every packet and encrypts either a null message $\top$ (chaff) or an application message. If the requested length cannot hold a nonce and tag, the sender emits uniform bytes and the receiver returns $\top$, so short traffic has no semantic output in either world. We use this design as our atomic layer (\S\ref{sec:datagram}) and add bidirectional semantics above it.

\section{Model, Scope, and Goals}
\label{sec:model}

\subsection{Endpoints, epochs, and schedules}
Two trusted endpoints \(A\) and \(B\) share direction-separated symmetric keys established before the FEP session begins. Time is divided into logical \emph{epochs}. In each epoch $t$, each endpoint $X\in\{A,B\}$ accepts at most one application input, emits scheduled traffic toward its peer $\overline X$, processes traffic received from $\overline X$, and may report visible close. Let $\mathcal M=\{0,1\}^*$ be the application-message space. The input is $\mu\in\mathcal M\cup\{\bot,\mathsf{closeReq}\}$. An input $m\in\mathcal M$ is application data and may be the empty message $\epsilon$; $\bot$ means that the application supplies no event in the epoch; and $\mathsf{closeReq}$ requests a local half-close. The close request is an application-level event that causes the protocol to generate an encrypted FIN. 

A public schedule $\Gamma_X(t)$ fixes the number of protocol-payload bytes $X$ emits while live: for datastreams, a byte-stream prefix that TCP may segment arbitrarily; for datagrams, the exact payload length of one UDP datagram. The pair $(\Gamma_A,\Gamma_B)$ need not be identical, and is preferred to be divergent and randomized for anti-fingerprinting.

The schedule may be fixed or sampled at session setup, but its generator must be independent of application content, sizes, queue occupancy, FIN, authentication failures, and every other session secret. Sampling a fresh schedule per session avoids a single constant pattern and permits a distribution designed to resemble cover traffic~\cite{li2018measuring,holland2022regulator,holland2024detorrent,witwer2022padding}. It does not by itself prevent fingerprinting: lengths, directions, packet counts, and timing remain observable, and the schedule distribution may itself be distinctive. Uniform-looking payloads can also form a detectable traffic class~\cite{wu2023great}. Epoch boundaries and the generator profile are fixed public protocol parameters; $L_{\mathrm{BD}}$ below records the realized per-session lengths and close outputs. Our theorems establish payload security conditioned on these public values. Selecting a statistically resistant schedule distribution is separate (\S\ref{sec:related}).

Let $\mathcal E_{\mathrm{cl}}$ denote the public set of epochs at which visible closure is permitted. The close grid may be fixed in advance or sampled during session setup, provided that its distribution is independent of application data and private protocol state. If no permitted epoch remains after an endpoint becomes ready, its visible-close output is $\bot$. Unlike the per-direction emission schedules $\Gamma_A$ and $\Gamma_B$, the grid is shared: an emission schedule governs one direction, whereas visible closure terminates the session as a whole. Thus, endpoints that become ready within the same epoch can close simultaneously.

Endpoint-specific grids $\mathcal E^A_{\mathrm{cl}}\neq\mathcal E^B_{\mathrm{cl}}$ would fit the definitions unchanged, since the close output of Eq.~\eqref{eq:streamclose} is endpoint-local, but they would make one-directional tails structural: between the two closure epochs one endpoint is silent while the other pays its full schedule, and in the datagram case an earlier closure can strand its peer (\S\ref{sec:eval-loss}).

For endpoint $X$, the close leakage $e_X^\star$ is $\bot$ if $X$ never becomes ready or no later close bucket exists. Otherwise it is fixed by the selected construction's public close rule: the datastream closes at the first epoch in $\mathcal E_{\mathrm{cl}}$ at or after readiness, whereas the datagram construction defers close by the fixed public \emph{linger depth} $L\in\mathbb N$ of further permitted buckets (\S\ref{sec:datagram}). The endpoints may realize different close epochs when readiness occurs on opposite sides of a bucket boundary because of asymmetric FIN drain or unequal delivery of FIN/ACK evidence (\S\ref{sec:eval-loss}).

\subsection{Adversary}

The adversary observes both directions and may delay, drop, fragment, coalesce, insert, delete, replay, reflect, reorder, and modify traffic arbitrarily. It sees the full transport/network wire image: addresses, ports, TCP/UDP headers, lengths, directions, timing, packet counts, and transport termination. The FEP claim covers exactly the protocol-controlled payload bytes, conditioned on the declared leakage
\begin{equation}
 L_{\mathrm{BD}}=(\Gamma_A,\ \Gamma_B,\ \mathcal E_{\mathrm{cl}},\
 e_A^\star,\ e_B^\star),
 \label{eq:leakage}
\end{equation}
where $e_X^\star$ is $X$'s realized close bucket (or $\bot$). TLS, QUIC, and other deployed transports intentionally expose selected protocol structure on the wire~\cite{rfc8446,rfc9000,rfc9312} and do not aim to provide the FEP guarantees studied here. 

\subsection{Goals}
\label{sec:properties}

We require:
(i)~\emph{correctness}: reliable in-order datastream delivery of admitted data; for datagrams, authentic DATA accepted before peer FIN is delivered atomically and at most once, while loss and reordering across FIN are availability events;
(ii)~\emph{shaping}: exact public output lengths in both directions in every live epoch, including idle ones;
(iii)~\emph{passive security}: live protocol bytes indistinguishable from uniform;
(iv)~\emph{protocol-state integrity}: only authentic peer actions may produce DATA delivery or advance the close state; forged DATA, forged FIN/ACK, DATA after FIN, and premature close are rejected;
(v)~\emph{private half-close and scheduled close}: half-close invisible; full close visible only in $\mathcal E_{\mathrm{cl}}$.

An active adversary may mount a denial-of-service attack by dropping or modifying traffic, thereby delaying delivery, stalling a direction, or preventing termination. Such interference does not violate integrity unless it causes forged data, control events, or close transitions to be accepted. Accordingly, our security guarantees exclude availability while preserving authenticity under denial of service.

\subsection{Security notions}
\label{sec:defs}

A bidirectional protocol is a tuple $\Pi=(\Init,\Send_A,\Send_B,\Recv_A,\Recv_B)$. With security parameter $\lambda$, $\Init(1^\lambda,\Gamma_A,\Gamma_B,\mathcal E_{\mathrm{cl}})$ returns two endpoint states; $L$ and all session bounds are fixed public protocol parameters. A send call $\Send_X(t,\mathsf{st}_X,\mu)$ returns an updated state, one scheduled ciphertext of length $\Gamma_X(t)$, and a local close bit. A receive call $\Recv_X(t,\mathsf{st}_X,c)$ returns an updated state, a delivery output (a list for datastreams and $x$ or $\bot$ for datagrams), and a close bit. In each epoch the sends run first, each receive uses its endpoint's post-send state, and only the combined close bits are applied absorbingly afterward (Algorithm~\ref{alg:epoch}).

\begin{definition}[Passive BiFEP security]
\label{def:passive}
After fixing the public schedules and close grid, a PPT adversary adaptively
supplies both endpoints' $\Send$\ inputs at each epoch. Let $V_0$ be its view of the honest
execution. In $V_1$, the same leakage $L_{\mathrm{BD}}$ is
shown, but $X$'s epoch-$t$ output is an independent uniform string of length
$\Gamma_X(t)$ when $e_X^\star=\bot$ or $t\le e_X^\star$, and is $\epsilon$
afterward. The protocol is passively secure if
every PPT adversary distinguishes $V_0$ from $V_1$ with only negligible
advantage. See Appendix~\ref{app:games}, Algorithm~\ref{alg:app-passive-epoch},
for the full experiment.
\end{definition}

\begin{definition}[Active BiFEP security]
\label{def:active}
The active game extends Definition~\ref{def:passive} by letting the adversary
modify and deliver traffic in both directions. The ideal world processes only
honest sender output: a datastream direction stops after its first deviation,
whereas each datagram is checked atomically. The event $\mathsf{Bad}$ records
accepted forged or replayed DATA/FIN/ACK, DATA after FIN, nonempty datastream
output after deviation, or premature close. We call $\Pi$ actively secure if
both the distinguishing advantage and
$\Pr[\mathsf{Bad}=1\mid b=0]$ are negligible, where $b=0$ denotes the real
experiment. Appendix~\ref{app:games}
specifies the full oracles and integrity monitor
(Algorithms~\ref{alg:app-active-stream-game}--\ref{alg:app-active-dg-recv}).
\end{definition}

The key difference from~\cite{fenske2024bytes} is who controls output lengths. In the unidirectional games, the adversary requests each length, so the definition protects byte contents but not application-dependent length patterns or per-channel close behavior. In Definition~\ref{def:passive}, lengths instead follow the public schedule until each leaked full-close epoch, jointly constraining both directions and hiding earlier local half-closes.

\begin{proposition}[Naive-composition separation]
\label{prop:naive}
Two channels can each satisfy the unidirectional FEP definitions while their
parallel composition fails passive BiFEP security.
\end{proposition}

\begin{proof}
Let $\Gamma_A(t)=\Gamma_B(t)=1024$ and
$\mathcal E_{\mathrm{cl}}=\{8,16,\ldots\}$. After sending its request in
epoch~2, $A$ half-closes while $B$'s response continues through epoch~6;
both are ready to close by epoch~8. A naive pair of unidirectional FEPs may
emit nothing from $A\to B$ after epoch~2, which their individual games permit
because the requested length is zero in both worlds. Testing
$|c_{A\to B}(3)|=0$ therefore reveals $A$'s half-close. In contrast,
Definition~\ref{def:passive} requires 1,024 bytes in both directions through
the leaked epoch-8 close bucket, hiding the event.
\end{proof}

\section{Bidirectional Datastream Construction}
\label{sec:stream}

We define the bidirectional datastream construction
\[
\Pi_{\mathrm{BD}}
=
(\Init,\Send_A,\Send_B,\Recv_A,\Recv_B)
\]
under the syntax of \S\ref{sec:model} and Appendix~\ref{app:games}. For each direction $X\to Y$, the construction applies an authenticated wrapper around an independently keyed unidirectional datastream FEP. The wrapper encodes application data, cover traffic, and close control as encrypted objects; the inner FEP converts the resulting object stream into ciphertexts of the requested length.

The wrapper uses a nonce-based AEAD scheme $(\Enc,\Dec)$ satisfying IND\$-CPA and INT-CTXT security. Let $\nu$ be an injective map from 64-bit
sequence numbers to AEAD nonces. We use
\[
\Enc_{ak}(s,m):=\Enc_{ak}(\nu(s),m)
\]
as shorthand, and define $\Dec_{ak}(s,\cdot)$ analogously. The wrapper sequence $s$ and inner record counter $r$ are bounded public session counters; admissible sessions end before either value repeats.

The inner layer is the unidirectional datastream FEP
$\Pi_{\mathrm{UD}}$ of Eq.~\eqref{eq:inner}, with sender $\mathsf{UD.Enc}$ and receiver $\mathsf{UD.Recv}$. We instantiate it with a trivial close function
\(
C_{\mathrm{UD}}\equiv 0\), so the inner layer never produces a close event. In particular, FIN processing and visible termination are governed entirely by the bidirectional wrapper, while the inner layer continues to emit the length requested by the public schedule.

Let $L_{\mathrm{in}}$ and $L_{\mathrm{type}}$ denote, in bytes, the encoded-length and frame-type widths; let $L_{\max}$ be the maximum application-chunk length, $B_{\mathrm{rec}}$ the maximum inner-record body length, and $c_{\mathrm{aead}}$ and $c_{\mathrm{enc}}$ the AEAD and inner-record overheads. Define $c_{\mathrm{wrap}}=L_{\mathrm{in}}+L_{\mathrm{type}}+c_{\mathrm{aead}}$. We require
\begin{equation}
c_{\mathrm{wrap}}\le B_{\mathrm{rec}},\qquad
B_{\mathrm{rec}}-L_{\mathrm{in}}<2^{8L_{\mathrm{in}}},
\label{eq:dummyfit}
\end{equation}
so an authenticated empty DUMMY fits in one record and every protected-object length is representable.

\subsection{State and initialization}
\label{sec:stream-state}

\begin{definition}[Endpoint state]
\label{def:streamstate}
For direction $X\to Y$, the send state is
\[
\begin{aligned}
 \sigma^S_{X\to Y}=\bigl(&k,\,ak,\,r,\,s^S,\,\mathsf{buf},\,\mathsf{obuf},\\
 &\mathsf{out},\,\mathsf{finrem},\,\mathsf{finlim},\,\mathsf{wst}\bigr),
\end{aligned}
\]
where $k$ and $ak$ are the inner and wrapper keys; $r$ and $s^S$ are the next inner record and wrapper sequence numbers; $\mathsf{buf}$ holds protected wrapper plaintext not yet placed in an inner record; $\mathsf{obuf}$ holds generated inner ciphertext not yet released; $\mathsf{out}\in\mathbb N$ counts released ciphertext bytes; $\mathsf{finrem},\mathsf{finlim}\in\mathbb N\cup\{\bot\}$ track the FIN position; and $\mathsf{wst}\in\{\mathsf{open},\mathsf{finQueued},\mathsf{finDrained}\}$ is the write state. The receive state for the opposite direction is
\[
 \sigma^R_{Y\to X}=\bigl(\mathsf{st}_{\mathrm{UD}},\,ak,\,s^R,\,
 \mathsf{ibuf},\,\mathsf{rfin}\bigr),
\]
where $\mathsf{st}_{\mathrm{UD}}$ is the inner receiver state, $s^R$ the expected wrapper sequence, $\mathsf{ibuf}$ the reassembly buffer, and $\mathsf{rfin}$ the peer-FIN bit. Endpoint $X$'s state is $(\sigma^S_{X\to Y},\sigma^R_{Y\to X},\mathsf{closed}_X,\mathsf{bad}_X)$ with the absorbing visible-close bit $\mathsf{closed}_X$ and the private failure bit $\mathsf{bad}_X$; $X$ is live in epoch $t$ if $\mathsf{closed}_X=0$ when the epoch begins.
\end{definition}

$\Init(1^\lambda,\Gamma_A,\Gamma_B,\mathcal E_{\mathrm{cl}})$ samples the four keys $k_{A\to B},ak_{A\to B},k_{B\to A},ak_{B\to A}$ independently and uniformly, sets every counter and bit to zero, every buffer to $\epsilon$, $\mathsf{wst}=\mathsf{open}$, and $\mathsf{finrem}=\mathsf{finlim}=\bot$. No key or state variable is shared across directions or layers.

\subsection{Protected objects}
\label{sec:stream-objects}

A wrapper frame is $F=\tau\Vert x$ with $\tau \in \{\DATA,\DUMMY,\FIN\}$: DATA carries one application chunk, DUMMY carries uniformly random cover bytes, and FIN has an empty payload. Protection of a frame at sequence number $s$ is defined as
\begin{equation}
 \begin{aligned}
 d&=\Enc_{ak_{X\to Y}}(s,\tau\Vert x),\\
 \Call{Protect}{\tau,x,s}
   &=\mathsf{BE}_{8L_{\mathrm{in}}}(|d|)\Vert d,
 \end{aligned}
 \label{eq:wrapper}
\end{equation}
after which the caller increments $s^S$. The sequence-derived nonce and the direction-separated key bind each object's type, payload, position, and direction, so replayed, reordered, and cross-direction-reflected objects fail authentication. The length prefix is itself inner-layer plaintext: on the wire it sits inside the inner encrypted datastream, so object boundaries are invisible.

\begin{figure}[t]
\centering
\begin{tikzpicture}[font=\scriptsize,
  layer/.style={draw, rounded corners=1.5pt, inner sep=3pt,
    text width=0.70\columnwidth, align=left, fill=blue!4},
  side/.style={font=\tiny\itshape, text=black!65, align=center,
    text width=0.16\columnwidth, fill=white, draw=black!35,
    rounded corners=1pt, inner sep=2pt},
  arr/.style={-{Stealth[length=5pt]}, semithick}]
\node[layer] (app) {\textbf{Application}\quad at most one input per epoch:
  $\mu\in\{\bot,\mathsf{closeReq}\}\cup\mathcal M$};
\node[layer, below=8pt of app] (wrap)
  {\textbf{Wrapper}\quad key $ak_{X\to Y}$, sequence $s$:\quad
   chunk to ${\le}L_{\max}$;\; $F=\tau\Vert x$,
   $\tau\in\{\DATA,\DUMMY,\FIN\}$;\;
   $P=\mathsf{BE}_{32}(|d|)\,\Vert\,d$, $d=\Enc_{ak}(s,F)$};
\node[layer, below=8pt of wrap] (lower)
  {\textbf{Inner FEP}\quad key $k_{X\to Y}$, record $r$:\quad
   body $u$ of ${\le}B_{\mathrm{rec}}$ queued wrapper bytes;\;
   $H{=}\Enc_k(2r,\mathsf{BE}_{16}(|C|))$,\;
   $C{=}\Enc_k(2r{+}1,\mathsf{BE}_{16}(0)\Vert u)$};
\node[layer, below=8pt of lower] (rel)
  {\textbf{Scheduled release}\quad ciphertext buffer $\mathsf{obuf}$;\;
   release exactly $\Gamma_X(t)$ bytes in epoch $t$, FIN drains only when its
   record fully leaves};
\node[layer, below=8pt of rel] (tcp)
  {\textbf{Transport (TCP)}\quad reads, writes, and packets carry no record,
   epoch, or object boundary};
\draw[arr] (app) -- (wrap);
\draw[arr] (wrap) -- (lower);
\draw[arr] (lower) -- (rel);
\draw[arr] (rel) -- (tcp);
\draw[decorate,decoration={brace,amplitude=3pt},semithick]
  ($(wrap.north east)+(2pt,0)$) -- ($(rel.south east)+(2pt,0)$);
\node[side, anchor=west]
  at ($(wrap.east)!0.5!(rel.east)+(9pt,0)$)
  {DUMMY fills\\any deficit};
\end{tikzpicture}
\caption{Sender-side layering for direction $X\to Y$. The reverse direction uses independent keys. Wrapper structure is inner-layer plaintext and is not visible on the wire.}
\label{fig:arch}
\end{figure}
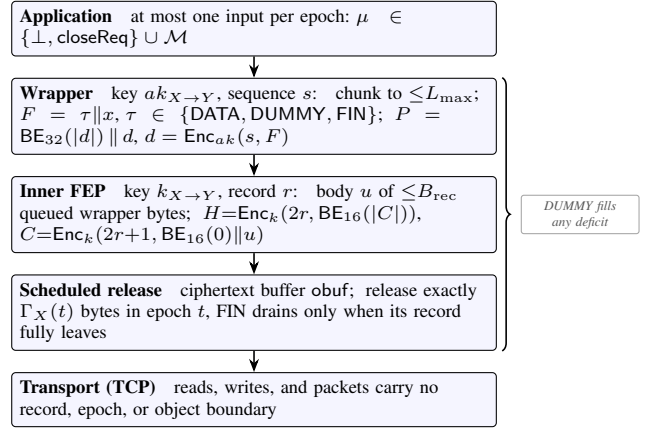

Application input is chunked canonically: $\Call{Chunk}{m}$ is the unique sequence $(m_1,\dots,m_j)$ with $m=m_1\Vert\cdots\Vert m_j$, $|m_i|=L_{\max}$ for $i<j$, and $0<|m_j|\le L_{\max}$; by convention $\Call{Chunk}{\epsilon}=(\epsilon)$, so an \emph{empty} DATA message is one authenticated object, distinct from supplying no input, which produces cover instead. To fill a ciphertext deficit $\delta$, the sender sizes one DUMMY as
\begin{align}
 \ell_d^{\max}&=\min\{L_{\max},\,B_{\mathrm{rec}}-c_{\mathrm{wrap}}\},\nonumber\\
 \ell_d&=\min\{\ell_d^{\max},
       \max\{0,\ \delta-c_{\mathrm{enc}}-c_{\mathrm{wrap}}\}\}.
 \label{eq:dummysize}
\end{align}
Overshooting a deficit is harmless: unused ciphertext stays in $\mathsf{obuf}$ for later epochs.

\subsection{Close state and the FIN-drain invariant}
\label{sec:stream-close}

Upon the first close request, the sender appends one protected FIN after queued DATA, enters $\mathsf{finQueued}$, and sets $\mathsf{finrem}$ to the number of inner-layer plaintext bytes through the FIN's final byte. Record generation decrements this counter by the bytes removed from $\mathsf{buf}$. For the record that consumes the final FIN byte, the sender sets
\begin{equation}
 \mathsf{finlim}=\mathsf{out}+|\mathsf{obuf}|+|d_{\mathrm{UD}}|,
 \label{eq:finlimit}
\end{equation}
where $d_{\mathrm{UD}}$ is that inner ciphertext record. Thus $\mathsf{finlim}$ is the absolute position of its final byte, and the sender enters $\mathsf{finDrained}$ only when $\mathsf{out}\ge\mathsf{finlim}$. A peer FIN sets only the private bit $\mathsf{rfin}$; it does not itself produce end-of-file, visible close, or a transport action.

Endpoint $X$ is \emph{ready} when its own FIN is drained, and its peer's FIN has been authenticated. The only close-related output is
\begin{equation}
 \mathsf{cl}_X(t)=\mathbf 1[
   \mathsf{wst}_{X\to Y}=\mathsf{finDrained}\ \land\
   \mathsf{rfin}_{Y\to X}\ \land\ t\in\mathcal E_{\mathrm{cl}}].
 \label{eq:streamclose}
\end{equation}
Figure~\ref{fig:close} shows the resulting behavior: both endpoints keep their full schedules through the closing epoch, and termination becomes visible only at the next public bucket after private readiness.

\begin{figure}[t]
\centering
\begin{tikzpicture}[font=\scriptsize, xscale=0.94]
\def\laneA{1.62} \def\laneB{0.62}
\draw[-{Stealth[length=5pt]}] (0.4,0) -- (8.8,0) node[below left=0pt and -14pt] {epoch $t$};
\foreach \e in {1,...,8} {\draw (\e,-0.05) -- (\e,0.05); \node[below] at (\e,-0.03) {\e};}
\foreach \e in {4,8} {
  \draw[dashed, black!55] (\e,-0.42) -- (\e,2.35);
  \node[black!55, font=\tiny, below] at (\e,-0.42) {$\in\mathcal E_{\mathrm{cl}}$};}
\node[left, font=\scriptsize] at (0.42,\laneA) {$A$};
\node[left, font=\scriptsize] at (0.42,\laneB) {$B$};
\fill[blue!22]  (0.55,\laneA-0.13) rectangle (4.5,\laneA+0.13);
\fill[black!12] (4.5,\laneA-0.13) rectangle (8.6,\laneA+0.13);
\fill[blue!22]  (0.55,\laneB-0.13) rectangle (4.5,\laneB+0.13);
\fill[black!12] (4.5,\laneB-0.13) rectangle (8.6,\laneB+0.13);
\node[font=\tiny, text=blue!50!black] at (2.5,\laneA-0.50)
  {scheduled bytes every epoch};
\node[font=\tiny, text=black!60] at (6.55,\laneA-0.50) {zero output};
\node[circle, fill=orange!80!black, inner sep=1.1pt,
      label={[font=\tiny]above:{\textsf{closeReq}}}] at (2,\laneA) {};
\node[circle, fill=orange!80!black, inner sep=1.1pt,
      label={[font=\tiny]below:{\textsf{closeReq}}}] at (2,\laneB) {};
\node[star, star points=5, fill=red!75!black, inner sep=1.2pt,
      label={[font=\tiny]above:{ready (private)}}] at (3.05,\laneA) {};
\node[star, star points=5, fill=red!75!black, inner sep=1.2pt,
      label={[font=\tiny]below:{ready (private)}}] at (3.05,\laneB) {};
\draw[-{Stealth[length=4pt]}, red!75!black] (4,\laneA+0.72) -- (4,\laneA+0.18);
\node[font=\tiny, text=red!75!black, above] at (4,\laneA+0.68) {visible close};
\draw[-{Stealth[length=4pt]}, red!75!black]
  (4,\laneB-0.48) -- (4,\laneB-0.18);
\node[font=\tiny, text=red!75!black, anchor=west]
  at (4.12,\laneB-0.43) {visible close};
\end{tikzpicture}
\caption{Close discipline on the baseline profile
($\mathcal E_{\mathrm{cl}}=\{4,8,\dots\}$). Both applications request close
in epoch~2; FINs drain and authenticate in epoch~3; the connection stays
byte-for-byte on schedule through epoch~4 and reveals close exactly there.}
\label{fig:close}
\end{figure}
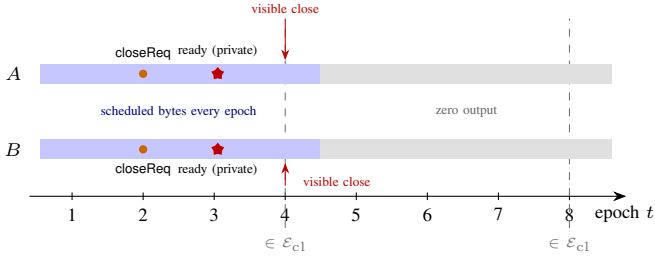

\subsection{Scheduled datastream sender}
\label{sec:stream-sender}

In epoch $t$, Algorithm~\ref{alg:stream-send} queues the wrapper object for input $\mu\in\{\bot,\mathsf{closeReq}\}\cup\mathcal M$, fills $\mathsf{obuf}$ with inner records, and releases $p=\Gamma_X(t)$ bytes.

\begin{algorithm}[t]
\caption{Close-aware datastream send at endpoint $X$}
\label{alg:stream-send}
{\footnotesize
\begin{algorithmic}[1]
\Require epoch $t$, input $\mu$, live directional state toward $Y$
\State $p\gets\Gamma_X(t)$
\If{$\mu=m\in\mathcal M$}
  \If{$\mathsf{wst}\ne\mathsf{open}$}
    \State $\mathsf{bad}\gets1$; reject $m$
  \Else
    \ForAll{$m_i\in\Call{Chunk}{m}$}
      \State $\mathsf{buf}\gets\mathsf{buf}\Vert
        \Call{Protect}{\DATA,m_i,s^S}$; $s^S\gets s^S+1$
    \EndFor
  \EndIf
\ElsIf{$\mu=\mathsf{closeReq}$ and $\mathsf{wst}=\mathsf{open}$}
  \State $P\gets\Call{Protect}{\FIN,\epsilon,s^S}$; $s^S\gets s^S+1$
  \State $\mathsf{finrem}\gets|\mathsf{buf}|+|P|$;
         $\mathsf{buf}\gets\mathsf{buf}\Vert P$
  \State $\mathsf{wst}\gets\mathsf{finQueued}$
\EndIf
\While{$|\mathsf{obuf}|<p$}
  \If{$\mathsf{buf}=\epsilon$}
    \State $\delta\gets p-|\mathsf{obuf}|$; compute $\ell_d$ by
           Eq.~\eqref{eq:dummysize}
    \State $x\gets\{0,1\}^{\ell_d}$ uniformly
    \State $\mathsf{buf}\gets\Call{Protect}{\DUMMY,x,s^S}$;
           $s^S\gets s^S+1$
  \EndIf
  \State $u\gets$ longest prefix of $\mathsf{buf}$ with
         $|u|\le B_{\mathrm{rec}}$; remove $u$
  \State $d_{\mathrm{UD}}\gets\mathsf{UD.Enc}_{k}(r,u)$; $r\gets r+1$
  \If{$\mathsf{finrem}\ne\bot$}
    \If{$|u|\ge\mathsf{finrem}$}
      \State $\mathsf{finrem}\gets\bot$;
      $\mathsf{finlim}\gets\mathsf{out}+|\mathsf{obuf}|+|d_{\mathrm{UD}}|$
    \Else
      \State $\mathsf{finrem}\gets\mathsf{finrem}-|u|$
    \EndIf
  \EndIf
  \State $\mathsf{obuf}\gets\mathsf{obuf}\Vert d_{\mathrm{UD}}$
\EndWhile
\State $c\gets\mathsf{obuf}[1..p]$; remove this prefix; $\mathsf{out}\gets\mathsf{out}+p$
\If{$\mathsf{finlim}\ne\bot$ and $\mathsf{out}\ge\mathsf{finlim}$}
  \State $\mathsf{finlim}\gets\bot$; $\mathsf{wst}\gets\mathsf{finDrained}$
\EndIf
\State \Return $(c,\mathsf{cl}_X(t))$ and updated state
\end{algorithmic}
}
\end{algorithm}

The following propositions give the shaping and FIN-drain properties used in Theorem~\ref{thm:stream}.

\begin{proposition}[Traffic shaping]
\label{prop:shaping}
Suppose Eq.~\eqref{eq:dummyfit} holds. For every epoch $t$, live endpoint $X$, input $\mu$, and reachable endpoint state, Algorithm~\ref{alg:stream-send} completes after finitely many iterations and returns a ciphertext string $c$ satisfying
\(
    |c|=\Gamma_X(t)
\).
\end{proposition}

\begin{proof}
Consider one iteration of the loop. If $\mathsf{buf}=\epsilon$, the sender creates one DUMMY whose protected length is $c_{\mathrm{wrap}}+\ell_d\le B_{\mathrm{rec}}$ by the definition of $\ell_d^{\max}$ and~\eqref{eq:dummyfit}, so the object enters a single record. Thus, whether $\mathsf{buf}$ already contains pending protected objects or is populated with a newly generated DUMMY object, the iteration removes a nonempty prefix $u$ of $\mathsf{buf}$ and appends $d_{\mathrm{UD}}$ with $|d_{\mathrm{UD}}|=|u|+c_{\mathrm{enc}}\ge 1+c_{\mathrm{enc}}$ to $\mathsf{obuf}$. Hence $|\mathsf{obuf}|$ strictly increases; the guard $|\mathsf{obuf}|<p$ fails after finitely many iterations, and the explicit prefix release returns exactly $p=\Gamma_X(t)$ bytes.
\end{proof}

\begin{proposition}[FIN-drain soundness]
\label{prop:findrain}
If $\mathsf{wst}_{X\to Y}=\mathsf{finDrained}$, then every byte of the inner record containing the FIN object has been returned by some earlier or current $\Send_X$ invocation.
\end{proposition}
\begin{proof}
When the inner record containing the final FIN byte is generated, its final ciphertext position is recorded as $\mathsf{finlim}=\mathsf{out}+|\mathsf{obuf}|+|d_{\mathrm{UD}}|$. Since records are released in order, the transition to $\mathsf{finDrained}$ occurs only when $\mathsf{out}\ge\mathsf{finlim}$, after the complete record has been released.
\end{proof}

Because $\mathsf{obuf}$ releases arbitrary ciphertext prefixes, epoch boundaries need not align with inner-record boundaries. Correctness likewise does not depend on the boundaries of transport reads, writes, or TCP segments.

\subsection{Datastream receiver and failure discipline}
\label{sec:stream-receiver}

Algorithm~\ref{alg:stream-recv} appends inner plaintext to $\mathsf{ibuf}$ and processes complete length-delimited objects in order, independent of input fragmentation; $\phi_{\mathrm{UD}}$ is the inner receiver's failure bit.

Let
\(
P=\mathsf{BE}_{8L_{\mathrm{in}}}(|d|)\Vert d
\) be the complete object at the head of $\mathsf{ibuf}$. At the expected sequence number $s^R$, the receiver accepts $P$ only if
\(
\Dec_{ak_{Y\to X}}(s^R,d)=\tau\Vert x
\) and
\(
\tau=\DUMMY\ \lor\
\bigl[\neg\mathsf{rfin}\land
 (\tau=\DATA\lor(\tau=\FIN\land x=\epsilon))\bigr]
\). Acceptance increments $s^R$ and applies the semantics of $\tau$. Failure removes $P$, sets $\mathsf{bad}$, and ends the call without incrementing $s^R$; the receiver never searches later byte offsets for another boundary.

\begin{algorithm}[t]
\caption{Close-aware datastream receive at endpoint $X$}
\label{alg:stream-recv}
{\footnotesize
\begin{algorithmic}[1]
\Require epoch $t$, incoming datastream fragment $c$, state from $Y$
\State $(\mathsf{st}_{\mathrm{UD}},z,\phi_{\mathrm{UD}})\gets
       \mathsf{UD.Recv}(\mathsf{st}_{\mathrm{UD}},c)$
\State $\mathsf{bad}\gets\mathsf{bad}\lor\phi_{\mathrm{UD}}$;
       $\mathsf{ibuf}\gets\mathsf{ibuf}\Vert z$; $\mathcal L\gets[\,]$
\While{$|\mathsf{ibuf}|\ge L_{\mathrm{in}}$}
  \State $q\gets\Call{DecodeBE}{\mathsf{ibuf}[1..L_{\mathrm{in}}]}$
  \State \textbf{if} $|\mathsf{ibuf}|<L_{\mathrm{in}}+q$ \textbf{then} \textbf{break} \Comment{retain incomplete object}
  \State remove $P=\mathsf{BE}_{8L_{\mathrm{in}}}(q)\Vert d$
         from the front of $\mathsf{ibuf}$
  \State $F\gets\Dec_{ak_{Y\to X}}(s^R,d)$
  \State \textbf{if} $F=\bot$ or $F$ is not a valid $\tau\Vert x$ \textbf{then} $\mathsf{bad}\gets1$; \textbf{break}
  \If{$\tau=\DATA$ and $\neg\mathsf{rfin}$}
    \State append $x$ to $\mathcal L$
  \ElsIf{$\tau=\DUMMY$}
    \State discard $x$
  \ElsIf{$\tau=\FIN$ and $x=\epsilon$ and $\neg\mathsf{rfin}$}
    \State $\mathsf{rfin}\gets1$
  \Else
    \State $\mathsf{bad}\gets1$; \textbf{break}
  \EndIf
  \State $s^R\gets s^R+1$
\EndWhile
\State \Return $(\mathcal L,\mathsf{cl}_X(t))$ and updated state
\end{algorithmic}
}
\end{algorithm}

Inner-receiver failures are permanently fail-stop because record alignment cannot be recovered; wrapper failures end only the current call. In both cases, $\mathsf{bad}$ remains private and affects neither output length nor visible close. Recovery requires a new higher-layer session.

\subsection{Epoch operation and absorbing close}
\label{sec:stream-epoch}

Algorithm~\ref{alg:epoch} runs both sends first. Each receive then uses its endpoint's post-send state; only the visible-close update is deferred. It combines $\mathsf{cl}_X\gets\mathsf{cl}_X^S\lor\mathsf{cl}_X^R$ and applies $\mathsf{closed}_X\gets\mathsf{closed}_X\lor\mathsf{cl}_X$ absorbingly. A closed endpoint subsequently ignores inputs and returns $(\epsilon,1)$. This ordering preserves all $\Gamma_X(t)$ bytes in the closing epoch and emits none afterward.

\begin{algorithm}[t]
\caption{Bidirectional epoch coordinator}
\label{alg:epoch}
{\footnotesize
\begin{algorithmic}[1]
\State $(c_{A\to B},\mathsf{cl}_A^S)\gets\Send_A(t,\mu_A)$
\State $(c_{B\to A},\mathsf{cl}_B^S)\gets\Send_B(t,\mu_B)$
\State adversary/transport produces deliveries $\hat c_{A\to B},
       \hat c_{B\to A}$
\State $(\mathcal L_B,\mathsf{cl}_B^R)\gets\Recv_B(t,\hat c_{A\to B})$
\State $(\mathcal L_A,\mathsf{cl}_A^R)\gets\Recv_A(t,\hat c_{B\to A})$
\State $\mathsf{cl}_A\gets\mathsf{cl}_A^S\lor\mathsf{cl}_A^R$;
       $\mathsf{cl}_B\gets\mathsf{cl}_B^S\lor\mathsf{cl}_B^R$
\State apply $\mathsf{cl}_A,\mathsf{cl}_B$ as absorbing state updates only now
\State \Return both scheduled outputs, deliveries, and close bits
\end{algorithmic}
}
\end{algorithm}

Cross-direction isolation is structural: by
Definition~\ref{def:streamstate}, $\sigma^S_{X\to Y}$ and $\sigma^R_{Y\to X}$ share no keys, counters, or buffers, and the only value computed from both is the close conjunction of Eq.~\eqref{eq:streamclose}, which is exactly the declared leakage.

\section{Bidirectional Datagram Construction}
\label{sec:datagram}

The companion construction $\Pi_{\mathrm{DG}}=(\Init,\Send_A,\Send_B,\Recv_A,\Recv_B)$ uses the same epoch interface and coordinator (Algorithm~\ref{alg:epoch}). Datagrams are atomic, may be reordered, duplicated, or lost, and provide no stream position from which to infer FIN delivery. 
As a result of this atomicity, the $\Send$ algorithm may return a message that is undeliverable, due to length or close state, to the application layer, an operation which we refer to as ``backpressure.''
The construction uses the nonce-based AEAD of \S\ref{sec:stream} with transmitted uniform nonces.
For fixed public parameters $L\in\mathbb N$ and $N_{\mathrm{sess}}\le2^{64}$, $\Init(1^\lambda,\Gamma_A,\Gamma_B,\mathcal E_{\mathrm{cl}})$ samples independent direction keys $k_{A\to B}$ and $k_{B\to A}$ and initializes each endpoint's semantic state: sequence $s^S=0$, replay sets $R_n=R_s=\emptyset$, readiness epoch $\rho=\bot$, and all bits of Eq.~\eqref{eq:dgstate} zero. Here $L$ is the linger depth of Eq.~\eqref{eq:dglinger}; sender algorithms abbreviate $k_{X\to Y}$ as $k$. We first define an atomic shaped channel, then the bidirectional semantics implemented by the wrapper.

\subsection{Atomic shaped channel}
\label{sec:dg-atomic}

Let $p\in[0,65507]$ be the requested UDP payload length (the IPv4 format bound), let $\ell_{\mathrm{nonce}}$ and $\ell_{\mathrm{tag}}$ be the AEAD nonce and tag lengths in bytes, and let $h=\ell_{\mathrm{nonce}}+\ell_{\mathrm{tag}}$. The channel has two sender maps, each returning exactly $p$ bytes, and one decoder. In both maps, $n\leftarrow\{0,1\}^{8\ell_{\mathrm{nonce}}}$ is a fresh uniform nonce, transmitted in the clear. For null input $\top$, $\mathsf{AtomicSendChaff}_k(p)$ outputs $p$ uniform bytes when $p<h+1$, and otherwise $n\Vert\Enc_k(n,0\Vert0^{p-h-1})$. For a message $m$ with $|m|\le2^{16}-1$ and $p\ge h+3+|m|$,
$\mathsf{AtomicSendData}_k(m,p)$ outputs $n\Vert\Enc_k(n,z)$ with plaintext
\begin{equation}
 z=1\Vert\mathsf{BE}_{16}(|m|)\Vert
 0^{p-h-3-|m|}\Vert m.
 \label{eq:dgbase}
\end{equation}
With AES-256-GCM, $h=28$: lengths 0--28 are unauthenticated uniform chaff, 29 bytes is the smallest authenticated null, and base DATA needs $31+|m|$ bytes. Although $65{,}507$ bytes is the maximum UDP payload size, the schedule generator should produce values $\Gamma_X(t)$ that avoid IP fragmentation along the intended network path~\cite{eggert2017rfc}.

Algorithm~\ref{alg:atomic-open} is the complete decoder $\mathsf{AtomicOpen}_k$. Padding precedes the message, so the authenticated 16-bit length selects the final $\ell$ bytes without putting $\ell$ on the wire.

\begin{algorithm}[t]
\caption{Atomic datagram receive $\mathsf{AtomicOpen}_k(c)\to(v,n)$}
\label{alg:atomic-open}
{\footnotesize
\begin{algorithmic}[1]
\State \textbf{if} $|c|<h+1$ \textbf{then} \Return $(\top,\bot)$ \Comment{short chaff}
\State \textbf{if} $|c|>65507$ \textbf{then} \Return $(\mathsf{invalid},\bot)$
\State parse $c=n\Vert d$; $z\gets\Dec_k(n,d)$
\State \textbf{if} $z=\bot$ or $|z|=0$ \textbf{then} \Return $(\mathsf{invalid},n)$
\State \textbf{if} $z[1]=0$ \textbf{then} \Return $(\top,n)$
\State \textbf{if} $z[1]\ne1$ or $|z|<3$ \textbf{then} \Return $(\mathsf{invalid},n)$
\State $\ell\gets\Call{DecodeBE}{z[2..3]}$
\State \textbf{if} $\ell>|z|-3$ \textbf{then} \Return $(\mathsf{invalid},n)$
\State \Return (the final $\ell$ bytes of $z$, $n$)
\end{algorithmic}
}
\end{algorithm}

Short and authenticated-null datagrams have no semantic output, so replaying them is harmless and consumes no anti-replay state. Authentication failure is local to one datagram: unlike a datastream, it creates no alignment problem for later packets.

\subsection{Semantic frame and endpoint state}

The non-null message given to the atomic channel is the semantic frame
\begin{equation}
 M=f\Vert\mathsf{BE}_{64}(s)\Vert x,
 \label{eq:dgframe}
\end{equation}
where the flag byte $f$ has fixed $\DATA$, $\FIN$, and $\ACK$ bit positions, $s$ is the sender's semantic sequence number, and $x$ is the application payload (possibly $\epsilon$). Algorithm~\ref{alg:dg-recv} enforces the valid flag combinations. Flags, sequence, and payload are all encrypted.
Endpoint $X$'s semantic state is its outgoing sequence number $s^S$, exact sets $R_n,R_s$ of accepted semantic nonces and sequences, the readiness epoch $\rho\in\mathbb N\cup\{\bot\}$ recorded when Eq.~\eqref{eq:dgclose} first holds, and monotone bits
\begin{equation}
 (\mathsf{ownFin},\mathsf{peerFin},\mathsf{ownFinAck},
   \mathsf{peerAckSent},\mathsf{closed},\mathsf{bad}).
 \label{eq:dgstate}
\end{equation}
Endpoint $X$ is \emph{ready}, written $\mathsf{ready}_X$, exactly when
\begin{equation}
 \mathsf{ownFin}_X\land\mathsf{peerFin}_X\land
 \mathsf{ownFinAck}_X\land\mathsf{peerAckSent}_X.
 \label{eq:dgclose}
\end{equation}
Valid semantic sequence numbers lie in $\{0,\ldots,N_{\mathrm{sess}}-1\}$. Once $s^S=N_{\mathrm{sess}}$, inputs requiring a new semantic frame receive backpressure and the sender emits scheduled chaff; previously established close bits may still advance through received traffic. An authenticated frame with $s\ge N_{\mathrm{sess}}$ is invalid and yields no semantic output. Thus exhaustion prevents further DATA or control transmission but does not erase readiness already obtained from existing evidence.

The exact sets $R_n$ and $R_s$ admit a frame only when both its nonce and sequence number are fresh, irrespective of arrival order.

The semantic wrapper adds nine bytes (one flag byte and eight sequence bytes), so the smallest semantic datagram is $28+3+9=40$ bytes and the largest application payload at public length $p$ is $p-40$. Input that does not fit receives backpressure without changing endpoint state, and $\Send_X$ emits exactly $p$ bytes of chaff.

\subsection{Datagram sender}

Algorithm~\ref{alg:dg-send} accepts at most one application input per epoch. After local half-close it repeats FIN, and after peer FIN it repeats ACK, within the existing schedule and through the linger phase of Eq.~\eqref{eq:dglinger}.

\begin{algorithm}[t]
\caption{Scheduled datagram send at endpoint $X$}
\label{alg:dg-send}
{\footnotesize
\begin{algorithmic}[1]
\Require epoch $t$, input $\mu$, live endpoint state
\State $p\gets\Gamma_X(t)$; $x\gets\epsilon$; $\mathsf{haveData}\gets0$
\If{$\mu=m\in\mathcal M$}
  \If{$\mathsf{ownFin}$}
    \State $\mathsf{bad}\gets1$; reject $m$
  \ElsIf{$p<40$ or $|m|>p-40$ or $s^S\ge N_{\mathrm{sess}}$}
    \State backpressure $m$
  \Else
    \State $x\gets m$; $\mathsf{haveData}\gets1$ \Comment{atomic admission}
  \EndIf
\ElsIf{$\mu=\mathsf{closeReq}$ and $\neg\mathsf{ownFin}$}
  \State \textbf{if} $p<40$ or $s^S\ge N_{\mathrm{sess}}$ \textbf{then} backpressure close \textbf{else} $\mathsf{ownFin}\gets1$
\EndIf
\State $\mathsf{need}\gets(\mathsf{haveData}\lor\mathsf{ownFin}
       \lor\mathsf{peerFin})\land s^S<N_{\mathrm{sess}}$
\If{$\mathsf{need}$ and $p\ge40$}
  \State $f\gets0$
  \State \textbf{if} $\mathsf{haveData}$ \textbf{then} $f\gets f\lor\DATA$
  \State \textbf{if} $\mathsf{ownFin}$ \textbf{then} $f\gets f\lor\FIN$
  \State \textbf{if} $\mathsf{peerFin}$ \textbf{then} $f\gets f\lor\ACK$
  \State $M\gets f\Vert\mathsf{BE}_{64}(s^S)\Vert x$
  \State $c\gets\mathsf{AtomicSendData}_k(M,p)$; $s^S\gets s^S+1$
  \State \textbf{if} $f$ contains ACK \textbf{then} $\mathsf{peerAckSent}\gets1$
\Else
  \State $c\gets\mathsf{AtomicSendChaff}_k(p)$
\EndIf
\State \textbf{if} $\rho=\bot$ and $\mathsf{ready}_X$ \textbf{then} $\rho\gets t$
\State \Return $(c,\mathsf{cl}_X(t))$ and updated state
\end{algorithmic}
}
\end{algorithm}

DATA and FIN are mutually exclusive because DATA is rejected after $\mathsf{ownFin}$ is set. DATA+ACK permits continued sending before local half-close; FIN+ACK carries both close signals. With no applicable flag, the sender emits chaff, so it never generates a zero flag byte.

\subsection{Datagram receiver, replay, and close}

After atomic authentication of a non-null message, Algorithm~\ref{alg:dg-recv} checks both the transmitted nonce and hidden sequence, rejecting repeated ciphertexts and duplicate semantic positions while allowing reordering.

\begin{algorithm}[t]
\caption{Datagram receive at endpoint $X$}
\label{alg:dg-recv}
{\footnotesize
\begin{algorithmic}[1]
\Require epoch $t$, received atomic datagram $c$ from $Y$
\State $(v,n)\gets\mathsf{AtomicOpen}_{k_{Y\to X}}(c)$
\State \textbf{if} $v=\top$ \textbf{then} \Return $(\bot,\mathsf{cl}_X(t))$
\State \textbf{if} $v=\mathsf{invalid}$ \textbf{then} $\mathsf{bad}\gets1$; \Return $(\bot,\mathsf{cl}_X(t))$
\State \textbf{if} $n\in R_n$ \textbf{then} \Return replay with no output
\State insert $n$ into $R_n$
\State \textbf{if} $|v|<9$ \textbf{then} $\mathsf{bad}\gets1$; \Return no output
\State parse $v=f\Vert\mathsf{BE}_{64}(s)\Vert x$
\State \textbf{if} $s\ge N_{\mathrm{sess}}$ \textbf{then} $\mathsf{bad}\gets1$; \Return no output
\State \textbf{if} $s\in R_s$ \textbf{then} \Return replay with no output
\State insert $s$ into $R_s$
\If{$f=0$ or $f$ has reserved bits or DATA+FIN}
  \State $\mathsf{bad}\gets1$; \Return no output
\EndIf
\If{$f$ has ACK and $\neg\mathsf{ownFin}$}
  \State $\mathsf{bad}\gets1$; \Return no output
\EndIf
\If{$f$ has DATA and $\mathsf{peerFin}$}
  \State $\mathsf{bad}\gets1$; \Return no output
\EndIf
\State \textbf{if} $f$ has FIN \textbf{then} $\mathsf{peerFin}\gets1$
\State \textbf{if} $f$ has ACK \textbf{then} $\mathsf{ownFinAck}\gets1$
\State \textbf{if} $f$ has DATA \textbf{then} deliver $x$ atomically \textbf{else} deliver nothing
\State \textbf{if} $\rho=\bot$ and $\mathsf{ready}_X$ \textbf{then} $\rho\gets t$
\State \Return delivery and $\mathsf{cl}_X(t)$
\end{algorithmic}
}
\end{algorithm}

FIN and ACK bits are monotone. ACK is valid only after local half-close, and DATA first arriving after authenticated peer FIN is rejected, including earlier DATA reordered across FIN. Such reordering is an availability loss, not acceptance of incorrect semantics. Authentication or semantic-validation failure discards only that datagram and sets private $\mathsf{bad}$; later datagrams remain independently processable.

$\mathsf{ownFinAck}$ records an authenticated acknowledgment of $X$'s FIN, while $\mathsf{peerAckSent}$ records that $X$ released an ACK for the peer's FIN. The send and receive algorithms record $\rho_X$ when Eq.~\eqref{eq:dgclose} first holds. The close output is
\begin{equation}
 \mathsf{cl}_X(t)=\mathbf 1\bigl[\mathsf{ready}_X\ \land\
   t\in\mathcal E_{\mathrm{cl}}\ \land\
   |\mathcal E_{\mathrm{cl}}\cap[\rho_X,t)|\ge L\bigr],
 \label{eq:dglinger}
\end{equation}
so visible close occurs at the $(L{+}1)$-st permitted bucket at or after readiness; $L=0$ closes at the first. While $s^S<N_{\mathrm{sess}}$, linger epochs with $\Gamma_X(t)\ge40$ emit fresh FIN{+}ACK frames; other epochs emit chaff. Algorithm~\ref{alg:epoch} applies the absorbing close only after same-epoch traffic is emitted and processed. Peer liveness requires acceptance of an eligible linger frame; loss, corruption, short schedules, or sequence exhaustion can prevent it.
Figure~\ref{fig:linger} contrasts the two disciplines on one adversarial execution: without linger the acknowledged endpoint's close makes the peer's missing ACK permanently undeliverable, while a one-bucket linger retransmits it past the drop window. Section~\ref{sec:eval-loss} evaluates the loss boundary.

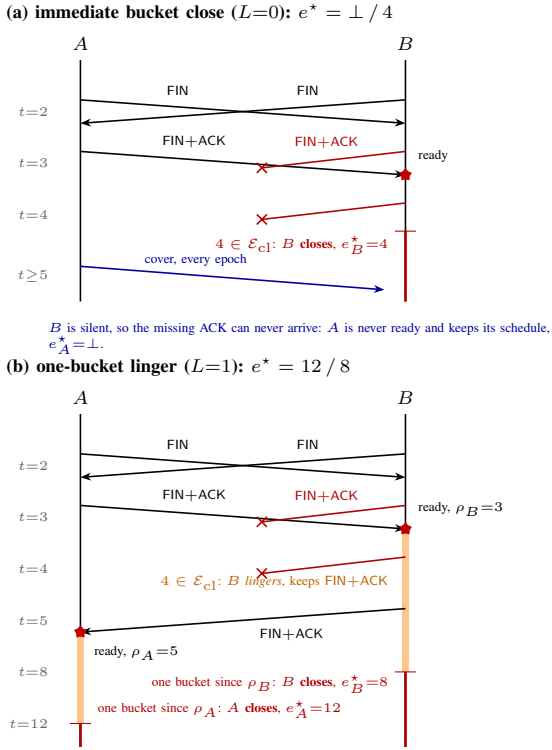
\begin{figure}[t]
\centering
\begin{tikzpicture}[font=\scriptsize, yscale=0.62]
\tikzset{msg/.style={-{Stealth[length=4pt]}, semithick},
         drop/.style={semithick, red!70!black},
         ep/.style={left, font=\tiny, text=black!60},
         ready/.style={star, star points=5, fill=red!75!black,
                       inner sep=1.2pt},
         closelab/.style={font=\tiny, text=red!70!black},
         lingerspan/.style={line width=2.6pt, orange!45}}
\def\Ax{1.15} \def\Bx{5.45}
\node[font=\scriptsize\bfseries, anchor=west] at (0.05,1.55)
  {(a) immediate bucket close ($L{=}0$): $e^\star=\bot\,/\,4$};
\draw[semithick] (\Ax,0.55) node[above] {$A$} -- (\Ax,-4.60);
\draw[semithick] (\Bx,0.55) node[above] {$B$} -- (\Bx,-4.60);
\node[ep] at (\Ax-0.30,-0.55) {$t{=}2$};
\draw[msg] (\Ax,-0.30) -- node[above=1pt, pos=0.30, font=\tiny] {$\FIN$} (\Bx,-0.80);
\draw[msg] (\Bx,-0.30) -- node[above=1pt, pos=0.30, font=\tiny] {$\FIN$} (\Ax,-0.80);
\node[ep] at (\Ax-0.30,-1.65) {$t{=}3$};
\draw[msg] (\Ax,-1.40) -- node[above=1pt, pos=0.35, font=\tiny] {$\FIN{+}\ACK$} (\Bx,-1.90);
\node[ready, label={[font=\tiny]above right:{ready}}] at (\Bx,-1.90) {};
\draw[drop] (\Bx,-1.40) -- node[above=1pt, pos=0.55, font=\tiny]
  {$\FIN{+}\ACK$} (3.55,-1.75) node[cross out, draw, inner sep=1.6pt] {};
\node[ep] at (\Ax-0.30,-2.75) {$t{=}4$};
\draw[drop] (\Bx,-2.50) -- (3.55,-2.85) node[cross out, draw, inner sep=1.6pt] {};
\node[closelab, align=right, left=3pt]
  at (\Bx,-3.40) {$4\in\mathcal E_{\mathrm{cl}}$: \textbf{$B$ closes}, $e_B^\star{=}4$};
\draw[line width=1pt, red!70!black] (\Bx,-3.10) -- (\Bx,-4.60);
\draw[red!70!black] (\Bx-0.14,-3.10) -- (\Bx+0.14,-3.10);
\node[ep] at (\Ax-0.30,-4.05) {$t{\ge}5$};
\draw[msg, blue!60!black] (\Ax,-3.85) --
  node[above=0pt, pos=0.38, font=\tiny, text=blue!60!black]
  {cover, every epoch} (\Bx-0.28,-4.35);
\node[font=\tiny, align=left, anchor=north west, text=blue!60!black,
      text width=7.0cm] at (\Ax-0.55,-4.85)
  {$B$ is silent, so the missing ACK can never arrive: $A$ is never ready
   and keeps its schedule, $e_A^\star{=}\bot$.};
\begin{scope}[yshift=-7.55cm]
\node[font=\scriptsize\bfseries, anchor=west] at (0.05,1.55)
  {(b) one-bucket linger ($L{=}1$): $e^\star=12\,/\,8$};
\draw[semithick] (\Ax,0.55) node[above] {$A$} -- (\Ax,-6.55);
\draw[semithick] (\Bx,0.55) node[above] {$B$} -- (\Bx,-6.55);
\draw[lingerspan] (\Bx,-1.90) -- (\Bx,-4.95);
\draw[lingerspan] (\Ax,-4.10) -- (\Ax,-6.05);
\node[ep] at (\Ax-0.30,-0.55) {$t{=}2$};
\draw[msg] (\Ax,-0.30) -- node[above=1pt, pos=0.30, font=\tiny] {$\FIN$} (\Bx,-0.80);
\draw[msg] (\Bx,-0.30) -- node[above=1pt, pos=0.30, font=\tiny] {$\FIN$} (\Ax,-0.80);
\node[ep] at (\Ax-0.30,-1.65) {$t{=}3$};
\draw[msg] (\Ax,-1.40) -- node[above=1pt, pos=0.35, font=\tiny] {$\FIN{+}\ACK$} (\Bx,-1.90);
\node[ready, label={[font=\tiny]above right:{ready, $\rho_B{=}3$}}]
  at (\Bx,-1.90) {};
\draw[drop] (\Bx,-1.40) -- node[above=1pt, pos=0.55, font=\tiny]
  {$\FIN{+}\ACK$} (3.55,-1.75) node[cross out, draw, inner sep=1.6pt] {};
\node[ep] at (\Ax-0.30,-2.75) {$t{=}4$};
\draw[drop] (\Bx,-2.50) -- (3.55,-2.85) node[cross out, draw, inner sep=1.6pt] {};
\node[font=\tiny, align=right, text=orange!75!black, left=3pt] at (\Bx,-3.05)
  {$4\in\mathcal E_{\mathrm{cl}}$: $B$ \emph{lingers}, keeps $\FIN{+}\ACK$};
\node[ep] at (\Ax-0.30,-3.85) {$t{=}5$};
\draw[msg] (\Bx,-3.60) -- node[below=1pt, pos=0.35, font=\tiny] {$\FIN{+}\ACK$} (\Ax,-4.10);
\node[ready, label={[font=\tiny]below right:{ready, $\rho_A{=}5$}}]
  at (\Ax,-4.10) {};
\node[ep] at (\Ax-0.30,-4.95) {$t{=}8$};
\node[closelab, align=right, left=3pt]
  at (\Bx,-5.20) {one bucket since $\rho_B$: \textbf{$B$ closes}, $e_B^\star{=}8$};
\draw[line width=1pt, red!70!black] (\Bx,-4.95) -- (\Bx,-6.55);
\draw[red!70!black] (\Bx-0.14,-4.95) -- (\Bx+0.14,-4.95);
\node[ep] at (\Ax-0.30,-6.05) {$t{=}12$};
\node[closelab, align=left, right=3pt]
  at (\Ax,-5.75) {one bucket since $\rho_A$: \textbf{$A$ closes}, $e_A^\star{=}12$};
\draw[line width=1pt, red!70!black] (\Ax,-6.05) -- (\Ax,-6.55);
\draw[red!70!black] (\Ax-0.14,-6.05) -- (\Ax+0.14,-6.05);
\end{scope}
\end{tikzpicture}
\caption{One-sided ACK-phase loss on the baseline profile
($\mathcal E_{\mathrm{cl}}=\{4,8,\dots\}$, close requests in epoch~2,
$B$-to-$A$ ACK-bearing epochs 3--4 dropped). (a)~Under immediate bucket
close, $B$ closes at epoch~4 and falls silent, so $A$ never authenticates an
ACK and stays live indefinitely. (b)~With a one-bucket linger, $B$ stays on
schedule through epoch~8, its epoch-5 FIN{+}ACK completes $A$'s evidence,
and both endpoints close (Table~\ref{tab:loss}, $k{=}2$ column).}
\label{fig:linger}
\end{figure}

Three invariants follow from the algorithms. \emph{Traffic shaping}: every live epoch emits one datagram of $\Gamma_X(t)$ bytes. \emph{At-most-once semantics}: $R_n$ and $R_s$ admit each semantic datagram's DATA and control meaning at most once despite reordering or duplication. \emph{Close safety}: Eq.~\eqref{eq:dgclose} requires the four FIN/ACK bits except with the atomic channel's forgery probability (Lemma~\ref{lem:dgclose}, Appendix~\ref{app:construction}), and Eq.~\eqref{eq:dglinger} can only defer visible close.

\begin{table*}[t]
\caption{Mechanisms of the unidirectional FEPs of~\cite{fenske2024bytes} and
of this work. Columns are paired so that each construction sits beside the
unidirectional channel it builds on; ``---'' means the notion does not arise.}
\label{tab:models}
\centering
\scriptsize
\setlength{\tabcolsep}{3.5pt}
\begin{tabular}{p{0.125\textwidth}p{0.185\textwidth}p{0.205\textwidth}p{0.175\textwidth}p{0.205\textwidth}}
\toprule
 & \multicolumn{2}{c}{Datastream} & \multicolumn{2}{c}{Datagram} \\
\cmidrule(lr){2-3}\cmidrule(lr){4-5}
 & Unidirectional~\cite{fenske2024bytes}
 & This work (\S\ref{sec:stream})
 & Unidirectional~\cite{fenske2024bytes}
 & This work (\S\ref{sec:datagram}) \\
\midrule
Directions & one & two& one & two\\
Unit & byte prefix & byte prefix & atomic datagram & atomic datagram \\
Delivery & reliable, ordered & reliable, ordered & loss/reorder/duplicate & loss/reorder/duplicate \\
Nonce & implicit counter & implicit counter & transmitted fresh & transmitted fresh \\
Cover & inner-layer padding & authenticated DUMMY object & null message $\top$ & null $\top$ shaped to $\Gamma_X(t)$ \\
\midrule
Frame typing & none (lengths only) & authenticated $\DATA/\DUMMY/\FIN$ & none & authenticated $\DATA/\FIN/\ACK$ flags \\
Semantic sequence & --- & wrapper counter $s$ & --- & hidden 64-bit counter \\
Replay defense & datastream position & datastream position & not required by FEP-CCA & exact nonce and sequence sets \\
Half-close & --- & private; cover continues & --- & private; cover continues \\
Close evidence & sender-side $C_{\mathrm{UD}}$ & FIN record released $\wedge$ peer FIN & none & FIN acknowledged $\wedge$ ACK sent \\
Visible close & sender-chosen & public bucket $\mathcal E_{\mathrm{cl}}$ & --- & public bucket $\mathcal E_{\mathrm{cl}}$, $L$-bucket linger \\
Failure & fail-stop & fail-stop (inner); call-local (wrapper) & discard one packet & discard one packet \\
Cross-direction isolation & --- & structural (four keys) & --- & structural (two keys) \\
Minimum cover & one byte & one requested byte & zero-byte payload & zero-byte payload \\
Minimum semantics & amortized & amortized & --- & 40-byte payload \\
\bottomrule
\end{tabular}
\end{table*}

\section{Security Analysis}
\label{sec:security}

We analyze the constructions under Definitions~\ref{def:passive}
and~\ref{def:active} of \S\ref{sec:defs}, with leakage $L_{\mathrm{BD}}$ from
Eq.~\eqref{eq:leakage}. The passive game replaces each live output with an
equal-length uniform string; the active game adds adversarial delivery to both
receivers, with prefix synchronization for streams and atomic matching for
datagrams. The monotone wrapper-integrity game covers forged DATA or FIN,
post-FIN DATA, and early close; Appendix~\ref{app:games} gives its concrete
bound.

\begin{theorem}[Datastream security]
\label{thm:stream}
Assume the inner unidirectional datastream FEP is correct and actively secure with
public output lengths and zero inner close, the wrapper AEAD has pseudorandom
ciphertexts (IND\$-CPA) and INT-CTXT security, the four direction/layer keys
are independent, wrapper and inner sequence numbers never repeat within a
session, and failure state is private. Then the construction of
\S\ref{sec:stream} provides bidirectional traffic shaping, correctness,
authenticated wrapper-state integrity, and active BiFEP
security with leakage $L_{\mathrm{BD}}$.
\end{theorem}

\emph{Proof sketch.}
Hybridize each direction's inner channel to uniform strings; the public
schedule fixes every hybrid's lengths (Proposition~\ref{prop:shaping}), and
all wrapper state (types, sequence numbers, FIN) is inner-layer plaintext,
hence hidden. Half-close timing is likewise hidden because drain is a
release-side event (Proposition~\ref{prop:findrain}). For active traffic, an
accepted out-of-sync DATA or FIN implies either an inner active-security break
or a wrapper AEAD forgery at the still-expected sequence number. Absent such
an event, the wrapper state machine is deterministic in honest inputs, so
post-FIN DATA is rejected and close cannot occur before
Eq.~\eqref{eq:streamclose} holds. A union bound over the two directions
yields the concrete bound in Appendix~\ref{app:games}. \hfill$\square$

\begin{theorem}[Datagram security]
\label{thm:datagram}
Assume the atomic datagram channel is correct and FEP-CCA secure, the two datagram
direction keys are independent of each other and of all datastream keys,
transmitted nonces repeat only with negligible probability, and the schedule
is application-independent. Then the construction of \S\ref{sec:datagram}
provides exact bidirectional datagram shaping, authentic atomic delivery before
peer FIN, replay-safe wrapper integrity, and active BiFEP security with leakage
$L_{\mathrm{BD}}$. Close \emph{liveness} additionally requires timely
acceptance of the required FIN/ACK evidence while both endpoints are live,
eligible control epochs and semantic sequence space remain, and a permitted
close bucket exists (Proposition~\ref{prop:linger}).
\end{theorem}

\emph{Proof sketch.}
Hybridize the two atomic channels independently; short chaff is null in both
worlds. A fresh accepted semantic frame not produced by the honest sender is
a FEP-CCA forgery. Nonce and hidden-sequence replay sets make every honest
frame effective at most once, and the FIN/ACK bits are monotone, so the close
predicate of Eq.~\eqref{eq:dgclose} cannot become true without an
authenticated peer FIN, an authenticated ACK of the endpoint's own FIN, and a
locally emitted ACK (Lemma~\ref{lem:dgclose},
Appendix~\ref{app:construction}). Loss may delay liveness and thereby change
the leaked close bucket.
The linger rule is determined by the readiness epoch and public pair
$(\mathcal E_{\mathrm{cl}},L)$; it changes only the realized bucket
$e_X^\star$, adds no further leakage, and cannot be advanced by the
adversary. \hfill$\square$

\begin{proposition}[Linger close liveness]
\label{prop:linger}
Suppose endpoint $X$ becomes ready at epoch $\rho_X$ and closes at
$e_X^\star$. If $Y$ accepts any fresh FIN{+}ACK frame emitted by $X$ in an
epoch $t\in[\rho_X,e_X^\star]$, then $Y$ becomes ready upon acceptance and its
close output follows Eq.~\eqref{eq:dglinger}. Consequently, $X$ can close while
$Y$ remains unready only if no eligible linger frame is accepted because of
loss or modification, a schedule below the 40-byte control minimum, or
semantic-sequence exhaustion.
\end{proposition}

\begin{proof}
Readiness of $X$ implies
$\mathsf{ownFin}_Y$, $\mathsf{peerFin}_Y$, and
$\mathsf{peerAckSent}_Y$; thus only $\mathsf{ownFinAck}_Y$ may be missing.
By Algorithm~\ref{alg:dg-send}, every semantic datagram emitted by $X$ from
$\rho_X$ to $e_X^\star$ carries FIN{+}ACK. Any fresh such datagram sets
$\mathsf{ownFinAck}_Y$, satisfying Eq.~\eqref{eq:dgclose}; closure then follows
from Eq.~\eqref{eq:dglinger}.
\end{proof}

The guarantees are conditional on the fixed public epoch timing and generator
profile and on $L_{\mathrm{BD}}$, which reveals realized lengths and close
buckets. Because the theorems are parametric in the generator, any
application-independent schedule and bucket distribution satisfies the same
payload guarantees, including fingerprinting-resistant designs
\cite{li2018measuring,holland2022regulator,holland2024detorrent,witwer2022padding,zhan2021website,siby2023evaluating}.
Selecting a distribution that resists statistical traffic analysis is outside
the games (\S\ref{sec:model}).

\section{Implementation}
\label{sec:implementation}


The reference implementation contains approximately 20.4k lines of Rust: 2,195 for the datastream construction, 989 for the datagram construction, 813 shared by both (the public-parameter layer and the AEAD boundary), and 16,375 of evaluation harness, with dependencies pinned by its lockfile. AES-256-GCM provides all AEAD operations, using 32-byte keys, 12-byte nonces, and 16-byte tags.
Datastream nonces are derived from nonrepeating, context-separated counters;
sessions end before counter exhaustion, preventing nonce reuse under a fixed
key. The evaluation harness derives keys deterministically for reproducibility;
production constructors obtain keys from the operating system's
cryptographically secure random-number generator.

\textbf{Datastream instantiation.}
Each direction uses an independent inner key and wrapper key (four keys
total). Wrapper parameters are $L_{\mathrm{in}}=4$, $L_{\mathrm{type}}=1$,
$L_{\max}=1024$, $B_{\mathrm{rec}}=4096$; with the 16-byte wrapper tag and
the 36-byte inner overhead of Eq.~\eqref{eq:inner}, a protected zero-payload
DUMMY costs 21 bytes, and the smallest scheduled epoch is one byte. The
baseline evaluation profile uses the asymmetric schedule
$\Gamma_A{=}1200$, $\Gamma_B{=}1000$ bytes/epoch with close buckets every
fourth epoch over a 64-epoch evaluation horizon.

\textbf{Datagram instantiation.}
The datagram construction adds two direction-specific keys, for six keys when
both constructions are instantiated. Each datagram transmits a freshly
sampled 12-byte nonce, giving $h=28$, a minimum authenticated-null size of 29
bytes, a minimum semantic-datagram size of 40 bytes, and $p-40$ application
bytes at public length $p$. The configured session limit is
$N_{\mathrm{sess}}=2^{32}$ semantic datagrams per direction. The datagram
evaluation reuses the datastream profile's baseline schedule and close grid.

The conformance, active, and passive datagram suites run at linger depth $L{=}0$, which reproduces immediate bucket close; the close-loss sweep of \S\ref{sec:eval-loss} sweeps $L\in\{0,\dots,3\}$.

\textbf{Resource policies.}
The implementation bounds sender queues and per-call receive input at 8\,MiB; oversized inputs produce backpressure without changing state or scheduled output. An incomplete wrapper object is limited to 1\,MiB and 64 receive calls, after which the receiver enters a permanent non-delivering state. DUMMY payloads are generated by ChaCha20~\cite{nir2018rfc} using operating-system entropy, except in the deterministic evaluation harness. 

\textbf{Adapters and coordination.}
For datastreams, TCP and timed adapters wrap the same endpoint objects; socket
reads and writes define neither epochs nor record boundaries. For datagrams,
the UDP adapter invokes one send per datagram. In each epoch, both endpoints
send, both process delivered peer traffic, and absorbing close is applied
last. The harness asserts this order in every trace. Instrumentation and
artifact I/O are excluded from byte accounting.

\section{Evaluation}
\label{sec:evaluation}

\subsection{Methodology and evidence classes}
\label{sec:eval-method}

Our evaluation studies the following questions.
\begin{itemize}
\item \textbf{RQ1 (conformance):} Do the implementations preserve transcripts, schedules, and close state under adversarial delivery?

\item \textbf{RQ2 (active integrity):} Can the tested modifications forge semantics, couple the directions, or force premature close?

\item \textbf{RQ3 (passive diagnostics):} Do finite byte statistics reveal an implementation artifact?

\item \textbf{RQ4 (cost):} What in-memory processing rate and cover expansion does the prototype exhibit, and what pacing error occurs on this host?
\end{itemize}

All experiments ran on Windows 11 (build 26200), an AMD Ryzen 7 7735HS with
16 logical CPUs, and Rust 1.94.0 (MSVC, release profile). 

Deterministic oracle tests assess protocol conformance; TCP/UDP loopback runs
check adapter byte accounting; passive statistical tests inspect byte-level
artifacts; and timed runs measure scheduler behavior on the evaluation host.
The cryptographic claims follow from the analysis of \S\ref{sec:security}.
\begin{table}[t]
\caption{Evaluation outcomes. Rows are non-additive; active rows are
subsets, and datagram active trials comprise 60 mutations plus 60 replays.}
\label{tab:harnesses}
\centering
\footnotesize
\setlength{\tabcolsep}{3pt}
\begin{tabular*}{\columnwidth}{@{\extracolsep{\fill}}lrrl@{}}
\toprule
Harness & Trials & Assertions & Outcome \\
\midrule
Stream conformance    & \nStreamTrials & \nStreamAssertionsTotal & passed \\
Stream mutations      & \nActiveStreamTrials & --- & 0 forged; target stalled \\
TCP mutations         & 30 pairs & 7,200 & 0 forged; target stalled \\
TCP length changes    & 6 pairs & 1,440 & 0 forged; target stalled \\
Timed TCP             & 150 & --- & oracle match \\
Datagram conformance  & \nDgCases{} cases & \nDgAssertions & passed \\
Datagram active       & 60+60 & --- & 0 forged; 0 redelivered \\
Parameter sweeps      & \nSweepPoints{} points & \nSweepAssertions & passed \\
Close-leak witness    & 3 pairs & 9 & naive 3/3; \sys{} 0/3 \\
\bottomrule
\end{tabular*}
\end{table}

\subsection{Datastream conformance and close (RQ1)}
\label{sec:eval-stream}

The deterministic suite comprised 143 scenarios, 491 trials, and
5,321 logical epochs; all 62,119 assertions passed
(Table~\ref{tab:streammatrix}). Correctness cases sweep application chunks
of 0--16{,}507 bytes across the $L_{\max}{=}1024$ and
$B_{\mathrm{rec}}{=}4096$ boundaries under direct delivery, one-byte and
boundary-straddling fragments, seeded random fragments, one-epoch delay,
and coalescing after a hold. All 10,642 per-direction live-epoch
measurements matched their scheduled lengths exactly.
\begin{table}[t]
\caption{Deterministic datastream conformance results. Normal/stall reports
sessions completing normally versus sessions in which the targeted direction
stalled; Fail counts assertion failures.}
\label{tab:streammatrix}
\centering
\footnotesize
\begin{tabular*}{\columnwidth}{@{\extracolsep{\fill}}lrrrr@{}}
\toprule
Category & Scenarios & Trials & Normal/stall & Fail \\
\midrule
Correctness/shape   & 54 & 54  & 54/0   & 0 \\
Exact schedule      & 43 & 43  & 43/0   & 0 \\
Close behavior      & 19 & 19  & 19/0   & 0 \\
Workloads           & 9  & 9   & 9/0    & 0 \\
Wrapper progression & 5  & 5   & 5/0    & 0 \\
Inner fail-stop     & 1  & 1   & 1/0    & 0 \\
Active integrity    & 12 & 360 & 0/360  & 0 \\
\bottomrule
\end{tabular*}
\end{table}

Nineteen close scenarios cover one-sided, simultaneous, staggered, and duplicate
close, sparse/dense buckets, backlog, FIN at record boundaries, post-FIN
writes, and a no-close control. Three invariants held throughout: a close
visible in epoch $t$ never shortened epoch $t$'s output; peer FIN caused
neither application EOF nor visible transport close; and closed endpoints
emitted zero bytes thereafter. Additional tests confirmed the failure
discipline of \S\ref{sec:stream}: wrapper failure is call-local and
non-advancing, whereas inner-record failure and resource quarantine are
permanently non-delivering.

\subsection{Active integrity (RQ2)}
\label{sec:eval-active}

To test active integrity, we ran a series of paired trials, in which an identical
application stream was either delivered normally or subjected to one 
or more active attacks, and the results were compared.  Across 360 datastream trials---insertion, deletion, truncation,
duplication, replay, reflection, injection before/within/after the stream,
single- and multi-bit flips, and FIN-adjacent flips, rotated across both
directions and simultaneous targeting---none delivered forged DATA or FIN or
caused premature close. In each single-direction mutation, the
reverse-direction workload completed unchanged as expected; the targeted direction could
stall.

The datastream suite replays the workload of~\cite{fenske2024bytes} using loopback
TCP sockets: 30 control/mutation pairs (2 directions $\times$ 5 target
emissions $\times$ 3 bit offsets). Each mutation caused authentication or
parsing failure in the targeted direction, while the other direction
completed. In six length-changing attacks (3-byte insertion, deletion, and
truncation in each direction), insertion and deletion failed authentication;
truncation left the receiver awaiting the missing suffix. 

Across 60 datagram mutation trials, bit flips, truncation, extension, and
uniform replacement prevented delivery of the target datagram, while datagrams in the
reverse direction and a later independent datagram always delivered. In 60
replay trials, the transmitted-nonce check rejected every previously
accepted ciphertext, producing no duplicate application output.

\subsection{Passive diagnostics (RQ3)}
\label{sec:eval-passive}

\begin{table}[t]
\caption{Passive byte diagnostics. Entropy is empirical byte entropy;
correlation is the maximum absolute byte correlation over lags 1--16 for
datastreams and lag 1 for datagrams. The 16-byte duplicate test applies only
to datastreams.}
\label{tab:passive}
\centering
\footnotesize
\setlength{\tabcolsep}{3pt}
\begin{tabular}{llrrr}
\toprule
 & Class & Entropy (bits/B) & $|$corr$|$ & 16-B dup. \\
\midrule
\multirow{4}{*}{\rotatebox{90}{DS}}
 & DATA A$\to$B & 7.999343 & 0.004071 & 0 \\
 & DATA B$\to$A & 7.999279 & 0.004668 & 0 \\
 & DUMMY A$\to$B & 7.999457 & 0.003168 & 0 \\
 & DUMMY B$\to$A & 7.999294 & 0.003974 & 0 \\
\midrule
\multirow{2}{*}{\rotatebox{90}{DG}}
 & DATA & 7.999749 & 0.001978 & --- \\
 & CHAFF & 7.999724 & 0.000219 & --- \\
\midrule
\multicolumn{2}{l}{Classifier} & ROC-AUC & \multicolumn{2}{r}{95\% CI} \\
\midrule
\multicolumn{2}{l}{datastream DATA vs.\ ChaCha20} & 0.485 &
 \multicolumn{2}{r}{[0.428, 0.542]} \\
\multicolumn{2}{l}{datastream DATA vs.\ DUMMY} & 0.459 &
 \multicolumn{2}{r}{[0.402, 0.516]} \\
\bottomrule
\end{tabular}
\end{table}

Table~\ref{tab:passive} reports byte-level diagnostics from 32 datastream
sessions of eight epochs per class and 64 datagram sessions. Empirical byte
entropy ranged from 7.999279 to 7.999749 bits/byte, and the largest absolute
correlation was 0.004668. No tested 16-byte datastream block repeated, and all
1,024 authenticated datagram nonces were distinct. The linear nearest-centroid
classifiers used session-disjoint even/odd training/test splits (256
observations per class in each split), byte-histogram, bit-balance,
lag-correlation, run, and compression-proxy features, and normal-approximation
intervals adjusted for two comparisons. Table~\ref{tab:passive} reports their
held-out ROC-AUC values; neither interval excludes 0.5. These finite-sample
diagnostics detected no byte-level implementation artifact and do not
establish computational indistinguishability.

\subsection{Processing cost and cover bandwidth (RQ4)}
\label{sec:eval-perf}

\begin{table}[t]
\caption{In-memory processing results (lower sample medians, 30
repetitions). Stream and datagram rows represent different operations and are
not a relative-overhead comparison.}
\label{tab:perf}
\centering
\footnotesize
\setlength{\tabcolsep}{3.2pt}
\begin{tabular}{lrrrr}
\toprule
Operation & App./op & Output/op & $\mu$s/op & Goodput \\
 & \multicolumn{2}{c}{(bytes)} & & (MiB/s) \\
\midrule
\sys{} stream epoch & \nPerfBdInputPerEpoch & \nPerfBdOutputPerEpoch & \nPerfBdLat & \nPerfBdGood \\
\dgsys{} atomic send/receive & 600 & 1,200 & \nDgPerfLatUs & \nDgPerfGood \\
\bottomrule
\end{tabular}
\end{table}

To test throughput and computational cost of our prototype, we ran an experiment using prebuilt 600/500-byte inputs and timed both endpoints' send and receive paths,
including cover generation, wrapper and lower-layer cryptography, parsing,
and state updates, excluding evaluation traces and assertions, network I/O,
and payload generation.  Across \nPerfBdEpochs{} live epochs, each operation
delivered \nPerfBdInputPerEpoch{} bytes, emitted exactly
\nPerfBdOutputPerEpoch{} scheduled bytes, and emptied both queues.
The timing results are shown in Table~\ref{tab:perf}, where App./op is admitted application data, Output/op is scheduled protocol output, and goodput is App./op divided by elapsed
operation time. For each stream operation,   The lower
median was \nPerfBdLat\,$\mu$s per epoch, \nPerfBdGood{} MiB/s of application
goodput, and \nPerfBdThpt{} MiB/s of shaped output. The distinct datagram
operation sends and receives one 600-byte message at a 1,200-byte public
length; it costs \nDgPerfLatUs\,$\mu$s and yields \nDgPerfGood{} MiB/s.

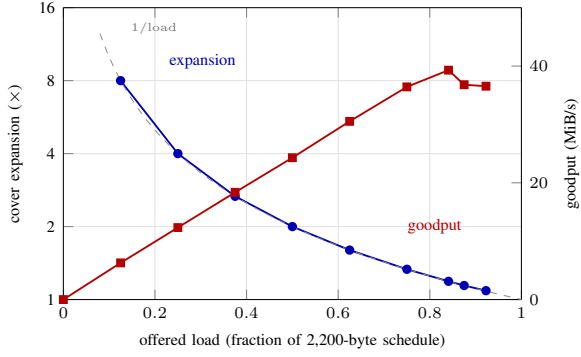
\begin{figure}[t]
\centering
\resizebox{0.9\columnwidth}{!}{%
\begin{tikzpicture}
\pgfplotsset{set layers}
\begin{semilogyaxis}[
  width=0.78\columnwidth, height=4.4cm,
  scale only axis,
  xlabel={offered load (fraction of 2{,}200-byte schedule)},
  ylabel={cover expansion ($\times$)},
  xmin=0, xmax=1, ymin=1, ymax=16,
  ytick={1,2,4,8,16},
  yticklabels={1,2,4,8,16},
  axis y line*=left,
  xlabel near ticks, ylabel near ticks,
  tick label style={font=\scriptsize},
  label style={font=\scriptsize},
  grid=major, grid style={black!12},
]
\addplot[blue!70!black, thick, mark=*, mark size=1.6pt]
  table[x=load, y=expansion] {biFep_NDSSD_Linger/figures/load_sweep.dat};
\addplot[black!45, dashed, domain=0.08:1, samples=60] {1/x};
\node[font=\tiny, text=black!55, anchor=west] at (axis cs:0.13,13.2) {$1/\mathrm{load}$};
\end{semilogyaxis}
\begin{axis}[
  width=0.78\columnwidth, height=4.4cm,
  scale only axis,
  xmin=0, xmax=1, ymin=0, ymax=50,
  axis y line*=right, axis x line=none,
  ylabel={goodput (MiB/s)},
  ylabel near ticks,
  tick label style={font=\scriptsize},
  label style={font=\scriptsize},
]
\addplot[red!70!black, thick, mark=square*, mark size=1.5pt]
  table[x=load, y=goodput] {biFep_NDSSD_Linger/figures/load_sweep.dat};
\end{axis}
\node[font=\scriptsize, text=blue!70!black] at (2.1,3.6) {expansion};
\node[font=\scriptsize, text=red!70!black] at (5.6,1.1) {goodput};
\end{tikzpicture}
}
\caption{Cover expansion and delivered goodput versus offered load on the
fixed 1,200/1,000-byte schedule (10 repetitions $\times$ 256 epochs). Each
epoch emits 2,200 bytes; expansion reaches \nLoadExpMaxUtil$\times$.}
\label{fig:load}
\end{figure}

\textbf{Cover expansion versus load.}
To test the ciphertext expansion of our prototype, we ran a series of trials sweeping the {\em offered load} -- application data supplied per epoch divided by the
scheduled traffic per epoch -- against our fixed 2,200-byte schedule; at nonzero load, 
we calculated expansion as the scheduled output divided by delivered
application data.
Figure~\ref{fig:load} shows the reults. At idle,
all 2,200 bytes/epoch are cover. Expansion is $2.0\times$ at half load and
falls to \nLoadExpMaxUtil$\times$ at the largest tested load; the excess
comprises wrapper/inner framing, schedule headroom, and packing slack. Every
point emitted the prescribed per-epoch lengths.


\textbf{Scaling with schedule size.}
For symmetric schedules $\Gamma_A=\Gamma_B=\Gamma$, a sweep from 200 to 12{,}800 bytes/epoch at 50\% load increased median epoch cost from \nSchedLatMin\,$\mu$s to \nSchedLatMax\,$\mu$s. Thus, a $64\times$ increase in scheduled bytes produced an approximately $13.9\times$ increase in median time, reducing measured time per scheduled byte.

\subsection{Close under loss (RQ1/RQ2)}
\label{sec:eval-loss}

\begin{table}[t]
\caption{Datagram close under loss. Both endpoints request close at epoch 2, with close buckets every four epochs; the indicated $k$ FIN/ACK epochs are dropped. Entries show first visible close epochs $A/B$, with $\bot$ denoting no close by epoch 32. Lower blocks repeat the patterns for $L=1,2,3$.}
\label{tab:loss}
\centering
\footnotesize
\setlength{\tabcolsep}{2pt}
\begin{tabular*}{\columnwidth}{@{\extracolsep{\fill}}lcccccc@{}}
\toprule
$k$ dropped & 1 & 2 & 3 & 4 & 5 & 6 \\
\midrule
\multicolumn{7}{@{}l}{\emph{Immediate bucket close ($L{=}0$)}} \\
FIN loss, symmetric
  & \nLossSymOne & \nLossSymTwo & \nLossSymThree
  & \nLossSymFour & \nLossSymFive & \nLossSymSix \\
FIN loss, one-sided
  & \nLossOneOne & \nLossOneTwo & \nLossOneThree
  & \nLossOneFour & \nLossOneFive & \nLossOneSix \\
ACK loss, one-sided
  & \nLossAckOne & \nLossAckTwo & \nLossAckThree
  & \nLossAckFour & \nLossAckFive & \nLossAckSix \\
\midrule
\multicolumn{7}{@{}l}{\emph{One-bucket linger ($L{=}1$)}} \\
FIN loss, symmetric
  & \nLossLingerSymOne & \nLossLingerSymTwo & \nLossLingerSymThree
  & \nLossLingerSymFour & \nLossLingerSymFive & \nLossLingerSymSix \\
FIN loss, one-sided
  & \nLossLingerOneOne & \nLossLingerOneTwo & \nLossLingerOneThree
  & \nLossLingerOneFour & \nLossLingerOneFive & \nLossLingerOneSix \\
ACK loss, one-sided
  & \nLossLingerAckOne & \nLossLingerAckTwo & \nLossLingerAckThree
  & \nLossLingerAckFour & \nLossLingerAckFive & \nLossLingerAckSix \\
\midrule
\multicolumn{7}{@{}l}{\emph{Two-bucket linger ($L{=}2$)}} \\
FIN loss, symmetric
  & \nLossLingerTwoSymOne & \nLossLingerTwoSymTwo & \nLossLingerTwoSymThree
  & \nLossLingerTwoSymFour & \nLossLingerTwoSymFive & \nLossLingerTwoSymSix \\
FIN loss, one-sided
  & \nLossLingerTwoOneOne & \nLossLingerTwoOneTwo & \nLossLingerTwoOneThree
  & \nLossLingerTwoOneFour & \nLossLingerTwoOneFive & \nLossLingerTwoOneSix \\
ACK loss, one-sided
  & \nLossLingerTwoAckOne & \nLossLingerTwoAckTwo & \nLossLingerTwoAckThree
  & \nLossLingerTwoAckFour & \nLossLingerTwoAckFive & \nLossLingerTwoAckSix \\
\midrule
\multicolumn{7}{@{}l}{\emph{Three-bucket linger ($L{=}3$)}} \\
FIN loss, symmetric
  & \nLossLingerThreeSymOne & \nLossLingerThreeSymTwo & \nLossLingerThreeSymThree
  & \nLossLingerThreeSymFour & \nLossLingerThreeSymFive & \nLossLingerThreeSymSix \\
FIN loss, one-sided
  & \nLossLingerThreeOneOne & \nLossLingerThreeOneTwo & \nLossLingerThreeOneThree
  & \nLossLingerThreeOneFour & \nLossLingerThreeOneFive & \nLossLingerThreeOneSix \\
ACK loss, one-sided
  & \nLossLingerThreeAckOne & \nLossLingerThreeAckTwo & \nLossLingerThreeAckThree
  & \nLossLingerThreeAckFour & \nLossLingerThreeAckFive & \nLossLingerThreeAckSix \\
\bottomrule
\end{tabular*}
\end{table}

Because datagram FIN/ACK evidence may be lost, we test whether loss can delay visible close without causing premature termination. Both endpoints request close at epoch~2, while an adversary drops the first $k$ FIN- or ACK-carrying epochs in one or both directions; we then record each endpoint's first visible close epoch.
FIN and ACK bits persist on every semantic datagram
(Algorithm~\ref{alg:dg-send}), and a ready endpoint retransmits them through
its linger window (Eq.~\eqref{eq:dglinger}).
Table~\ref{tab:loss} reports
$L\in\{0,\dots,3\}$; the no-loss controls ($k{=}0$) close at
\nLossSymZero{}, \nLossLingerSymZero{}, \nLossLingerTwoSymZero{}, and
\nLossLingerThreeSymZero{}, respectively. FIN-phase loss
delayed closure in public-bucket steps. The symmetric and one-sided FIN rows
coincide because the ACK exchange completes within one epoch after a FIN is
delivered.

With $L{=}0$, one-sided ACK loss let $B$ close while $A$ remained open: $A$
did not receive the ACK for its FIN before $B$ entered absorbing close. With
$L{=}1$, every tested pattern with $k\le5$ ended in bilateral close, with the
endpoints closing at most one bucket apart. At $k{=}6$, the one-sided ACK-loss
pattern suppressed every ACK-bearing datagram from $B$ before its epoch-8
close, leaving $A$ open
(Figure~\ref{fig:linger} traces the $k{=}2$ executions of the first two
blocks).
Each unit of $L$ delayed every otherwise successful close, including the
lossless control, by one public bucket, and the $L{=}2$ and $L{=}3$ blocks
left no endpoint open at any tested $k$: on this grid the acknowledged
endpoint closes at epoch $4L{+}4$, so the one-sided stall requires
$k\ge4L{+}2$, beyond the tested window for $L\ge2$. Every
observed close occurred at a permitted bucket after the four authenticated
conditions of Eq.~\eqref{eq:dgclose}; cases without bilateral close were
liveness failures permitted by Proposition~\ref{prop:linger}.

\subsection{Naive composition versus BiFEP}
\label{sec:eval-close-leak}

\begin{table}[t]
\caption{Half-close leakage. Close requests occur at the starts of epochs 2
and 6; silence onset is $A\to B\,/\,B\to A$. Witness reports successful
distinctions among three deterministic pairs.}
\label{tab:naive}
\centering
\footnotesize
\begin{tabular*}{\columnwidth}{@{\extracolsep{\fill}}lccc@{}}
\toprule
 & Silence onset & Half-close inference & Witness \\
\midrule
Naive composition
  & 2\,/\,6 & epochs 2 and 6 & 3/3 \\
\sys{} (this work)
  & 9\,/\,9 & none & 0/3 \\
\bottomrule
\end{tabular*}
\end{table}

We compare \sys{} with a naive composition of two unidirectional FEPs under the same staggered-close workload and measure whether per-direction silence reveals private half-close times. The sweep instantiates the sender-side witness of Proposition~\ref{prop:naive} with two independent inner-FEP senders. Each
meets its requested length while open and emits zero after local half-close
and drain. Close requests occur at epochs 2 and 6 on the baseline 1,200/1,000
schedule with grid $\{4,8,\dots\}$.

Table~\ref{tab:naive} gives the deterministic outcome: naive composition
exposes both half-close times and their order, whereas \sys{} keeps both
directions scheduled through the shared close bucket. Valid inner-FEP
ciphertext alone does not provide the joint idle-cover, half-close, and
coordinated-close properties of
Definitions~\ref{def:passive}--\ref{def:active}.

\subsection{Wall-clock scheduling fidelity}
\label{sec:eval-timed}

\begin{figure}[t]
\centering
\resizebox{0.9\columnwidth}{!}{%
\begin{tikzpicture}
\pgfplotsset{set layers}
\begin{axis}[
  width=0.75\columnwidth, height=4.2cm,
  scale only axis,
  xlabel={public epoch duration (ms)},
  ylabel={epochs started ${>}2$\,ms late (\%)},
  xmode=log, log ticks with fixed point,
  xmin=4, xmax=125, ymin=0, ymax=100,
  xtick={5,10,20,50,100},
  axis y line*=left,
  xlabel near ticks, ylabel near ticks,
  tick label style={font=\scriptsize},
  label style={font=\scriptsize},
  grid=major, grid style={black!12},
]
\addplot[blue!70!black, thick, mark=*, mark size=1.6pt]
  table[x=duration, y=missA] {biFep_NDSSD_Linger/figures/timed_sweep.dat};
\addplot[blue!45, thick, dashed, mark=o, mark size=1.6pt]
  table[x=duration, y=missB] {biFep_NDSSD_Linger/figures/timed_sweep.dat};
\addplot[black!60, thick, dotted, mark=diamond*, mark size=1.9pt]
  table[x=duration, y=overrun] {biFep_NDSSD_Linger/figures/timed_sweep.dat};
\node[font=\scriptsize, text=blue!70!black, anchor=west] at (axis cs:23,93) {late start ($A$, $B$)};
\node[font=\scriptsize, text=black!60, anchor=west] at (axis cs:5.6,30) {overrun};
\end{axis}
\begin{axis}[
  width=0.75\columnwidth, height=4.2cm,
  scale only axis,
  xmode=log, xmin=4, xmax=125, ymin=0, ymax=10,
  axis y line*=right, axis x line=none,
  ylabel={median start jitter (ms)},
  ylabel near ticks,
  tick label style={font=\scriptsize},
  label style={font=\scriptsize},
]
\addplot[red!70!black, thick, mark=square*, mark size=1.5pt]
  table[x=duration, y=jitter] {biFep_NDSSD_Linger/figures/timed_sweep.dat};
\node[font=\scriptsize, text=red!70!black, anchor=west] at (axis cs:55,4.2) {jitter};
\end{axis}
\end{tikzpicture}
}
\caption{Wall-clock fidelity on this host (30 trials $\times$ 10 epochs,
both endpoints; 2\,ms late-start threshold). No overrun was observed at 50
or 100\,ms; late starts occurred at every tested period.}
\label{fig:timed}
\end{figure}
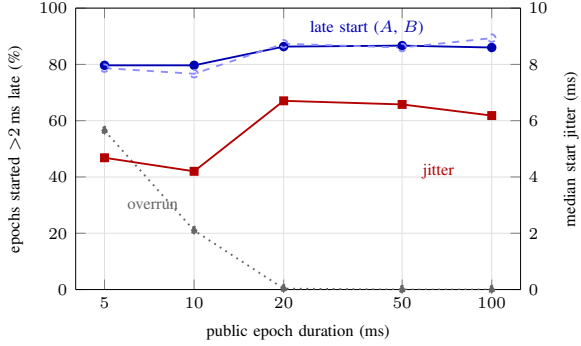

While our formal definition of traffic shaping is in terms of traffic per epoch,
a deployed implementation would adhere to a wall-clock schedule.  To measure
the ability of our prototype to meet wall-clock start times, we ran two independently clocked endpoint
tasks over loopback TCP at fixed periods of 5, 10, 20, 50, and 100\,ms. Start
jitter is actual minus nominal start time, so positive values denote delay. An epoch overruns if
it completes after the next nominal start. All 150 ten-epoch trials matched the logical
oracle, and late epochs were neither skipped nor merged. Across periods,
\nTimedMissRange\% of endpoint-epochs started more than 2\,ms late, and median
start jitter was \nTimedJitterRange\,ms. Overruns affected
\nTimedOverFive\% of 5\,ms epochs, \nTimedOverTen\% of 10\,ms epochs, and
0.33\% of 20\,ms epochs; none were observed at 50 or 100\,ms. These
measurements characterize only the tested host and scheduler.



\subsection{Comparison with deployed transports}
\label{sec:comparison}

\begin{table*}[t]
\caption{Specification-level property comparison, not an implementation
benchmark. ``Partial'' denotes optional, mode-dependent, or incomplete
support. The illustrative unit is protocol-specific, excludes IP/TCP/UDP
headers, and is not comparable across rows.}
\label{tab:comparison}
\centering
\scriptsize
\begin{tabular}{lcccccccccl}
\toprule
Protocol & Uniform & Active & DATA/ctl. & Bidir. & Idle & Hidden & Scheduled & Cross-dir. & Illustrative & Evidence \\
 & payload & integrity & indist. & shaping & cover & half-close & close & isolation & unit & \\
\midrule
TLS 1.3 & No & Yes & Partial & No & No & No & No & No & 22 B & \cite{rfc8446} \\
QUIC v1 & No & Partial & Partial & No & No & Partial & No & Partial & Initial $\ge$1200 B & \cite{rfc9000,rfc9312,thomson2021rfc} \\
WireGuard & No & Yes & No & No & No & N/A & N/A & Partial & 32 B keepalive & \cite{donenfeld2017wireguard} \\
Shadowsocks 2022 & Partial & Yes & Partial & No & No & No & No & Partial & $\ge$69 B & \cite{shadowsocks2022sip022} \\
obfs4 & Partial & Yes & Partial & No & No & No & No & No & $\ge44$ B (profile) & \cite{angel2014obfs4,fenske2024bytes} \\
Tor & No & Yes & Partial & No & Partial & Partial & No & No & 514-B cell & \cite{torproject2026specs} \\
Naive $2\times$UD-FEP & Yes & Yes & Yes & No & No & No & No & Yes & 1 requested B & \S\ref{sec:eval-close-leak} \\
\sys{} datastream & Yes & Yes & Yes & Yes & Yes & Yes & Yes & Yes & 1 requested B & \S\ref{sec:security} \\
\dgsys{} & Yes & Yes & Yes & Yes & Yes & Yes & Yes & Yes & 0 B chaff / 40 B ctl. & \S\ref{sec:security} \\
\bottomrule
\end{tabular}
\end{table*}

Table~\ref{tab:comparison} applies the properties of \S\ref{sec:properties}
at the specification level and includes naive composition as a calibration
row. Independent unidirectional FEPs provide per-direction passive security
and integrity, together with isolation from the independently keyed reverse
direction, but not the joint shaping and close properties measured in
Table~\ref{tab:naive}.

Deployed transports expose different public structure or make padding and
cover optional: TLS~1.3 retains outer record metadata and treats a bad record
MAC as connection-fatal~\cite{rfc8446}; QUIC exposes public header fields and
Initial-packet processing~\cite{rfc9000,rfc9312}; and WireGuard exposes
message types and counters~\cite{donenfeld2017wireguard}. Shadowsocks~2022
and obfs4 do not mandate idle cover in both directions or scheduled close
\cite{shadowsocks2022sip022,angel2014obfs4}, while Tor padding is
configurable~\cite{torproject2026specs}. No surveyed specification mandates
the conjunction of bidirectional shaping, idle cover in both directions, hidden
half-close, and scheduled close.

\section{Discussion and Limitations}
\label{sec:discussion}
\textbf{Deployment requirements.}
Our construction assumes authenticated shared keys and agreed public parameters;
as we discuss in section~\ref{sec:related}, other work has considered 
fully-encrypted key exchange and integrating such definitions with our 
work and prototype remains an interesting question for future work.
We note that given fresh shared session material and domain-separated derivation, the
endpoints can obtain a common sampled schedule and close grid without an
additional in-session message. Note that any in-band negotiation must itself satisfy
the declared leakage boundary.  Full deployments would need to additionally provide
congestion control, NAT traversal, fragmentation-safe sizing, and
application-data reliability within the public schedule.

\textbf{Costs and operating points.}
BiFEP’s schedule determines its bandwidth, latency, and processing costs. Provisioning capacity above the application load increases cover traffic (Figure~\ref{fig:load}), while longer scheduling intervals amortize processing overhead but may increase buffering delay (\S\ref{sec:eval-perf}). Very short intervals can also exceed the implementation’s processing capacity: on our evaluation host, we observed deadline overruns at 5, 10, and 20,ms, but none at 50 or 100,ms (Figure~\ref{fig:timed}). A deployment should therefore select its schedule according to expected load, latency requirements, and available processing capacity.

\textbf{Close latency versus liveness.}
The public datagram linger depth $L$ trades close latency for loss resilience: each increment delays visible close, including in lossless sessions, by one bucket of $\mathcal E_{\mathrm{cl}}$ and extends retransmission through the intervening eligible epochs. If $n$ eligible FIN+ACK datagrams in this window are lost independently with probability $q$, the probability that all are lost is $q^n$ (\S\ref{sec:eval-loss}). A public application-level timeout can bound resource use, but yields an abort rather than authenticated full close.

\textbf{Replay memory.}
Our implementation detects replays by storing every accepted nonce or sequence number for the session. This permits arbitrary packet reordering, but memory grows linearly with the number of accepted semantic datagrams. Although $N_{\mathrm{sess}}=2^{32}$ bounds the number of semantic frames per direction, it is not intended as a feasible size for this replay set. A deployment can instead use a bounded DTLS-style anti-replay window~\cite{rescorla2022rfc}, which limits memory at the cost of rejecting packets reordered beyond the configured window.

\section{Related Work}
\label{sec:related}

\textbf{FEPs and obfuscation transports.}
Our work builds directly on Fenske and Johnson’s definitions and unidirectional constructions \cite{fenske2024bytes}. Our new contributions are a bidirectional FEP model that captures both traffic directions and shared termination, a formal treatment of close and half-close behavior, and a datagram construction using authenticated FIN/ACK state. Deployed obfuscation transports, including obfs4, Shadowsocks, and ScrambleSuit, seek random-looking wire images but are not specified under the FEP formal definitions
\cite{angel2014obfs4,shadowsocks2022sip022,winter2013scramblesuit}. Prior work on obfuscated transports and their detectability shows that protocol mimicry and ad hoc randomization can leave classifier-visible artifacts~\cite{houmansadr2013parrot,wang2015seeing}. Other works have shown that application- and transport-level characteristics~\cite{xue2024fingerprinting,hanlon2024detecting,wang2025custom} and packet-level reactions~\cite{fifield2023comments} can lead to distinguishers, motivating our explicit treatment of lengths, reactions, and close behavior. 
G\"unther {\em et al.}~\cite{gunther2024obfuscated,gunther2025hybrid} have recently studied the obfuscation of key exchange and KEM-based public key encryption, but those works also do not explicitly treat message ordering, lengths, reactions, or close behavior.

\textbf{Channel security and traffic analysis.}
Previous work on bidirectional channel security~\cite{marson2017security}, security under ciphertext fragmentation~\cite{fischlin2015data,boldyreva2012security},and secure termination~\cite{boyd2017secure} studied state, fragmentation, and close behavior without requiring a random wire shape. Website-fingerprinting attacks and defenses instead select an emission schedule and quantify its cost~\cite{shmatikov2006timing,cai2014systematic,zhan2021website,siby2023evaluating,mathews2023sok}.
These defenses mitigate leakage intentionally left outside the BiFEP abstraction; the security of the resulting public schedule distribution must therefore be evaluated separately. TLS/QUIC wire-image management~\cite{rfc8446,rfc9000,rfc9312} and Encrypted Client Hello (ECH), which encrypts privacy-sensitive TLS handshake metadata~\cite{rfc9849,rfc9848}, reduce selected protocol-visible information but do not seek a uniform wire image. BiFEP is complementary: it targets leakage from the subsequent bidirectional traffic pattern and session termination.

\section{Conclusion}
\label{sec:conclusion}

We introduced BiFEPs, to the best of our knowledge, the first framework for fully encrypted bidirectional communication that treats the two traffic directions and session termination as a single security object. We formalized its security requirements and constructed both datastream and datagram BiFEPs, including mechanisms for traffic shaping, authenticated protocol state, and privacy-preserving termination.

We evaluated these constructions through a Rust implementation and a specification-level analysis of existing protocols. The implementation passed a suite of security tests, including adversarial mutation, dropping and replay tests, while achieving cover expansion from \(2.0\times\) at half load to \(\nLoadExpMaxUtil\times\) near saturation. We also examined existing deployed transports against the BiFEP requirements and found that none of the surveyed protocols satisfies the full property set.

Together, these results show that bidirectional fully encrypted communication requires more than independently protecting two unidirectional channels: traffic behavior and termination must be coordinated across the session. BiFEPs provides a formal model and a concrete realization of this stronger notion, combining bidirectional traffic shaping, active integrity, and private termination.

\section*{Ethical Considerations}
This work aims to protect network-session metadata against surveillance. Our evaluation used only synthetic traffic and systems under our control; it involved no human subjects, user data, or third-party networks. We acknowledge that there is the potential for harms as a result of the deployment of our protocols due to the dual-use nature of private communication technologies.  However, under the principal of beneficience, we believe that the  potential benefits of increased privacy and greater access to information, freedom of expression, and freedom of association outweigh these harms.

\bibliographystyle{IEEEtran}
\bibliography{biFep_NDSSD_Linger/ref}

\begin{thebibliography}{10}
\providecommand{\url}[1]{#1}
\csname url@samestyle\endcsname
\providecommand{\newblock}{\relax}
\providecommand{\bibinfo}[2]{#2}
\providecommand{\BIBentrySTDinterwordspacing}{\spaceskip=0pt\relax}
\providecommand{\BIBentryALTinterwordstretchfactor}{4}
\providecommand{\BIBentryALTinterwordspacing}{\spaceskip=\fontdimen2\font plus
\BIBentryALTinterwordstretchfactor\fontdimen3\font minus \fontdimen4\font\relax}
\providecommand{\BIBforeignlanguage}[2]{{%
\expandafter\ifx\csname l@#1\endcsname\relax
\typeout{** WARNING: IEEEtran.bst: No hyphenation pattern has been}%
\typeout{** loaded for the language `#1'. Using the pattern for}%
\typeout{** the default language instead.}%
\else
\language=\csname l@#1\endcsname
\fi
#2}}
\providecommand{\BIBdecl}{\relax}
\BIBdecl

\bibitem{rfc8446}
\BIBentryALTinterwordspacing
E.~Rescorla, ``{The Transport Layer Security (TLS) Protocol Version 1.3},'' RFC 8446, Aug. 2018. [Online]. Available: \url{https://www.rfc-editor.org/info/rfc8446}
\BIBentrySTDinterwordspacing

\bibitem{rfc9000}
\BIBentryALTinterwordspacing
J.~Iyengar and M.~Thomson, ``{QUIC: A UDP-Based Multiplexed and Secure Transport},'' RFC 9000, May 2021. [Online]. Available: \url{https://www.rfc-editor.org/info/rfc9000}
\BIBentrySTDinterwordspacing

\bibitem{rfc9312}
\BIBentryALTinterwordspacing
M.~K{\"u}hlewind and B.~Trammell, ``{Manageability of the QUIC Transport Protocol},'' RFC 9312, Sep. 2022. [Online]. Available: \url{https://www.rfc-editor.org/info/rfc9312}
\BIBentrySTDinterwordspacing

\bibitem{angel2014obfs4}
\BIBentryALTinterwordspacing
Y.~Angel, ``{obfs4 -- The obfourscator},'' accessed: Aug. 16, 2026. [Online]. Available: \url{https://github.com/Yawning/obfs4}
\BIBentrySTDinterwordspacing

\bibitem{shadowsocks2022sip022}
\BIBentryALTinterwordspacing
{Shadowsocks Contributors}, ``{SIP022: AEAD-2022 Ciphers},'' accessed: Aug. 16, 2026. [Online]. Available: \url{https://shadowsocks.org/doc/sip022.html}
\BIBentrySTDinterwordspacing

\bibitem{wu2023great}
M.~Wu, J.~Sippe, D.~Sivakumar, J.~Burg, P.~Anderson, X.~Wang, K.~Bock, A.~Houmansadr, D.~Levin, and E.~Wustrow, ``How the great firewall of china detects and blocks fully encrypted traffic,'' in \emph{32nd USENIX Security Symposium (USENIX Security 23)}, 2023, pp. 2653--2670.

\bibitem{wang2015seeing}
L.~Wang, K.~P. Dyer, A.~Akella, T.~Ristenpart, and T.~Shrimpton, ``Seeing through network-protocol obfuscation,'' in \emph{Proceedings of the 22nd ACM SIGSAC Conference on Computer and Communications Security}, 2015, pp. 57--69.

\bibitem{houmansadr2013parrot}
A.~Houmansadr, C.~Brubaker, and V.~Shmatikov, ``The parrot is dead: Observing unobservable network communications,'' in \emph{2013 IEEE Symposium on Security and Privacy}.\hskip 1em plus 0.5em minus 0.4em\relax IEEE, 2013, pp. 65--79.

\bibitem{fenske2024bytes}
E.~Fenske and A.~Johnson, ``Bytes to schlep? use a fep: Hiding protocol metadata with fully encrypted protocols,'' in \emph{Proceedings of the 2024 on ACM SIGSAC Conference on Computer and Communications Security}, 2024, pp. 1982--1996.

\bibitem{boyd2017secure}
C.~Boyd and B.~Hale, ``Secure channels and termination: The last word on tls,'' in \emph{International Conference on Cryptology and Information Security in Latin America}.\hskip 1em plus 0.5em minus 0.4em\relax Springer, 2017, pp. 44--65.

\bibitem{postel1980rfc0768}
J.~Postel, ``Rfc0768: User datagram protocol,'' 1980.

\bibitem{thomson2021rfc}
M.~Thomson and S.~Turner, ``Rfc 9001: Using tls to secure quic,'' 2021.

\bibitem{donenfeld2017wireguard}
J.~A. Donenfeld, ``Wireguard: next generation kernel network tunnel.'' in \emph{NDSS}, 2017, pp. 1--12.

\bibitem{torproject2026specs}
\BIBentryALTinterwordspacing
{The Tor Project}, ``Tor specifications,'' accessed 2026-07-29. [Online]. Available: \url{https://spec.torproject.org/}
\BIBentrySTDinterwordspacing

\bibitem{fischlin2015data}
M.~Fischlin, F.~G{\"u}nther, G.~A. Marson, and K.~G. Paterson, ``Data is a stream: Security of stream-based channels,'' in \emph{Annual Cryptology Conference}.\hskip 1em plus 0.5em minus 0.4em\relax Springer, 2015, pp. 545--564.

\bibitem{boldyreva2012security}
A.~Boldyreva, J.~P. Degabriele, K.~G. Paterson, and M.~Stam, ``Security of symmetric encryption in the presence of ciphertext fragmentation,'' in \emph{Annual International Conference on the Theory and Applications of Cryptographic Techniques}.\hskip 1em plus 0.5em minus 0.4em\relax Springer, 2012, pp. 682--699.

\bibitem{marson2017security}
G.~A. Marson and B.~Poettering, ``Security notions for bidirectional channels,'' \emph{IACR Transactions on Symmetric Cryptology}, pp. 405--426, 2017.

\bibitem{li2018measuring}
S.~Li, H.~Guo, and N.~Hopper, ``Measuring information leakage in website fingerprinting attacks and defenses,'' in \emph{Proceedings of the 2018 ACM SIGSAC Conference on Computer and Communications Security}, 2018, pp. 1977--1992.

\bibitem{holland2022regulator}
J.~K. Holland and N.~Hopper, ``Regulator: A straightforward website fingerprinting defense,'' \emph{Proceedings on Privacy Enhancing Technologies}, vol.~2, pp. 344--362, 2022.

\bibitem{holland2024detorrent}
J.~K. Holland, J.~Carpenter, S.~E. Oh, and N.~Hopper, ``Detorrent: An adversarial padding-only traffic analysis defense,'' \emph{Proceedings on Privacy Enhancing Technologies}, vol.~1, pp. 98--115, 2024.

\bibitem{witwer2022padding}
E.~Witwer, J.~K. Holland, and N.~Hopper, ``Padding-only defenses add delay in tor,'' in \emph{Proceedings of the 21st workshop on privacy in the electronic society}, 2022, pp. 29--33.

\bibitem{eggert2017rfc}
L.~Eggert, G.~Fairhurst, and G.~Shepherd, ``Rfc 8085: Udp usage guidelines,'' 2017.

\bibitem{zhan2021website}
P.~Zhan, L.~Wang, and Y.~Tang, ``Website fingerprinting on early quic traffic,'' \emph{Computer Networks}, vol. 200, p. 108538, 2021.

\bibitem{siby2023evaluating}
S.~Siby, L.~Barman, C.~Wood, M.~Fayed, N.~Sullivan, and C.~Troncoso, ``Evaluating practical quic website fingerprinting defenses for the masses,'' \emph{Proceedings on Privacy Enhancing Technologies}, 2023.

\bibitem{nir2018rfc}
Y.~Nir and A.~Langley, ``Rfc 8439: Chacha20 and poly1305 for ietf protocols,'' 2018.

\bibitem{rescorla2022rfc}
E.~Rescorla, H.~Tschofenig, and N.~Modadugu, ``Rfc 9147: The datagram transport layer security (dtls) protocol version 1.3,'' 2022.

\bibitem{winter2013scramblesuit}
P.~Winter, T.~Pulls, and J.~Fuss, ``Scramblesuit: A polymorphic network protocol to circumvent censorship,'' in \emph{Proceedings of the 12th ACM workshop on Workshop on privacy in the electronic society}, 2013, pp. 213--224.

\bibitem{xue2024fingerprinting}
\BIBentryALTinterwordspacing
D.~Xue, M.~Kallitsis, A.~Houmansadr, and R.~Ensafi, ``Fingerprinting obfuscated proxy traffic with encapsulated {TLS} handshakes,'' in \emph{USENIX Security Symposium}.\hskip 1em plus 0.5em minus 0.4em\relax USENIX, 2024. [Online]. Available: \url{https://www.usenix.org/system/files/sec24summer-prepub-465-xue.pdf}
\BIBentrySTDinterwordspacing

\bibitem{hanlon2024detecting}
\BIBentryALTinterwordspacing
M.~Hanlon, G.~Wan, A.~Ascheman, and Z.~Durumeric, ``Detecting {VPN} traffic through encapsulated {TCP} behavior,'' in \emph{Free and Open Communications on the Internet}, 2024. [Online]. Available: \url{https://www.petsymposium.org/foci/2024/foci-2024-0016.pdf}
\BIBentrySTDinterwordspacing

\bibitem{wang2025custom}
W.~Wang, D.~Xue, P.~Kumar, A.~Mishra, R.~Ensafi \emph{et~al.}, ``Is custom congestion control a bad idea for circumvention tools?'' \emph{Free and Open Communications on the Internet}, 2025.

\bibitem{fifield2023comments}
\BIBentryALTinterwordspacing
D.~Fifield, ``Comments on certain past cryptographic flaws affecting fully encrypted censorship circumvention protocols,'' Cryptology {ePrint} Archive, Paper 2023/1362, 2023. [Online]. Available: \url{https://eprint.iacr.org/2023/1362}
\BIBentrySTDinterwordspacing

\bibitem{gunther2024obfuscated}
\BIBentryALTinterwordspacing
F.~G{\"u}nther, D.~Stebila, and S.~Veitch, ``Obfuscated key exchange,'' in \emph{Proceedings of the 2024 on ACM SIGSAC Conference on Computer and Communications Security}, ser. CCS '24.\hskip 1em plus 0.5em minus 0.4em\relax New York, NY, USA: Association for Computing Machinery, 2024, p. 2385–2399. [Online]. Available: \url{https://doi.org/10.1145/3658644.3690220}
\BIBentrySTDinterwordspacing

\bibitem{gunther2025hybrid}
F.~G{\"u}nther, M.~Rosenberg, D.~Stebila, and S.~Veitch, ``Hybrid obfuscated key exchange and kems,'' in \emph{Annual International Cryptology Conference}.\hskip 1em plus 0.5em minus 0.4em\relax Springer, 2025, pp. 575--609.

\bibitem{shmatikov2006timing}
V.~Shmatikov and M.-H. Wang, ``Timing analysis in low-latency mix networks: Attacks and defenses,'' in \emph{European Symposium on Research in Computer Security}.\hskip 1em plus 0.5em minus 0.4em\relax Springer, 2006, pp. 18--33.

\bibitem{cai2014systematic}
X.~Cai, R.~Nithyanand, T.~Wang, R.~Johnson, and I.~Goldberg, ``A systematic approach to developing and evaluating website fingerprinting defenses,'' in \emph{Proceedings of the 2014 ACM SIGSAC conference on computer and communications security}, 2014, pp. 227--238.

\bibitem{mathews2023sok}
N.~Mathews, J.~K. Holland, S.~E. Oh, M.~S. Rahman, N.~Hopper, and M.~Wright, ``Sok: A critical evaluation of efficient website fingerprinting defenses,'' in \emph{2023 IEEE Symposium on Security and Privacy (SP)}.\hskip 1em plus 0.5em minus 0.4em\relax IEEE, 2023, pp. 969--986.

\bibitem{rfc9849}
\BIBentryALTinterwordspacing
E.~Rescorla, K.~Oku, N.~Sullivan, and C.~A. Wood, ``{TLS Encrypted Client Hello},'' RFC 9849, Mar. 2026. [Online]. Available: \url{https://www.rfc-editor.org/info/rfc9849}
\BIBentrySTDinterwordspacing

\bibitem{rfc9848}
\BIBentryALTinterwordspacing
B.~M. Schwartz, M.~Bishop, and E.~Nygren, ``{Bootstrapping TLS Encrypted ClientHello with DNS Service Bindings},'' RFC 9848, Mar. 2026. [Online]. Available: \url{https://www.rfc-editor.org/info/rfc9848}
\BIBentrySTDinterwordspacing

\end{thebibliography}
\appendices
\section*{Generative AI Assistance Disclosure}
Large language model assistants (OpenAI Codex and Anthropic Claude) assisted with paper polishing, implementation, and evaluation. The authors reviewed and are responsible for every claim, proof, and artifact.

\section*{Open Science}
The source code and artifacts required to reproduce our evaluation are
available in an anonymous repository at~\url{https://anonymous.4open.science/r/biFep-C235/README.md}.

\section{Formal Security Games and Concrete Bounds}
\label{app:games}

\subsection{Correctness}

Correctness has four parts. \emph{Datastream preservation} requires the
concatenated delivered chunks in each direction to equal the concatenated
admitted DATA. \emph{Eventual delivery} is conditioned on reliable in-order
delivery and continued live schedule invocations. \emph{Half-close
correctness} requires exactly one authenticated FIN after prior DATA and no
later DATA in that direction. \emph{Close correctness} permits visible close
only after local and peer half-close readiness and only in
$\mathcal E_{\mathrm{cl}}$.

For a datagram protocol, one send produces one atomic datagram and one
receive consumes one. Loss, duplication, and reordering are allowed.
An intact, honestly emitted DATA datagram first accepted before authenticated
peer FIN returns its exact payload atomically; honest control advances only
its encoded bits, and detected replays return null. Loss, duplication, or
reordering across FIN may suppress delivery but cannot change accepted
semantics. Every sender output respects the UDP payload bound. Close liveness
is conditioned on acceptance of the required FIN/ACK evidence while both
endpoints remain live and semantic sequence space remains
(\S\ref{sec:eval-loss}).
The linger phase of Eq.~\eqref{eq:dglinger} widens this live window: a ready
endpoint remains live, retransmitting FIN/ACK, through $L$ further public
close buckets (Proposition~\ref{prop:linger}).

\subsection{Passive challenge}

The following oracle realizes Definition~\ref{def:passive} without choosing
ideal lengths from secret real-world outputs. After fixing the public
parameters, the challenger runs $\Init(1^\lambda,\Gamma_A,\Gamma_B,
\mathcal E_{\mathrm{cl}})$ and maintains the resulting real shadow execution.
Calls below update the stored endpoint states in
place; we omit those state outputs for readability. Queries use strictly increasing epochs; both sends run
before both honest receives, and close is applied only after all four calls,
as in Algorithm~\ref{alg:epoch}. Let $\mathsf{live}_X^t$ record whether $X$
was live at the beginning of epoch $t$. For challenge bit $b$, the epoch
oracle is given in Algorithm~\ref{alg:app-passive-epoch}.

\begin{algorithm}[t]
\caption{Passive epoch oracle
$\mathcal O^{b}_{\mathsf{Epoch}}(t,\mu_A,\mu_B)$}
\label{alg:app-passive-epoch}
{\footnotesize
\begin{algorithmic}[1]
\For{$X\in\{A,B\}$}
  \State $(c^0_X,\mathsf{cl}_X^S)\gets\Send_X(t,\mathsf{st}_X,\mu_X)$
\EndFor
\State $(\cdot,\mathsf{cl}_B^R)\gets\Recv_B(t,\mathsf{st}_B,c^0_A)$
\State $(\cdot,\mathsf{cl}_A^R)\gets\Recv_A(t,\mathsf{st}_A,c^0_B)$
\For{$X\in\{A,B\}$}
  \State apply $\mathsf{cl}_X^S\lor\mathsf{cl}_X^R$ absorbingly
  \State \textbf{if} $b=0$ \textbf{then} $c^b_X\gets c^0_X$
  \State \textbf{if} $b=1\land\mathsf{live}_X^t$ \textbf{then}
         $c^b_X\gets\Rand(\Gamma_X(t))$
  \State \textbf{if} $b=1\land\neg\mathsf{live}_X^t$ \textbf{then}
         $c^b_X\gets\epsilon$
\EndFor
\State \Return $(c^b_A,c^b_B)$
\end{algorithmic}
}
\end{algorithm}
When $\mathcal A$ requests finalization, the challenger discloses the two
first absorbing close epochs $(e_A^\star,e_B^\star)$ (or $\bot$) before
obtaining its guess. Since
$\mathsf{live}_X^t=1$ exactly through $e_X^\star$, the random branch is a
function of $L_{\mathrm{BD}}$ alone and gives exactly $V_1$ of
Definition~\ref{def:passive}. In particular, it does not mirror a
secret-dependent real-output length.

\subsection{Active datastream challenge}

For a direction $d=X\to Y$, the active challenger stores the real shadow
sent stream $C^{0,S}_d$, the stream $C^S_d$ shown to $\mathcal A$, the stream
$C^R_d$ submitted for receipt, and a monotone synchronization bit
$\mathsf{sync}_d$. Algorithm~\ref{alg:app-active-stream-game} defines the
experiment.

\begin{algorithm}[t]
\caption{Active datastream experiment
$\mathsf{ActStr}^{b}_{\Pi,\mathcal A}$}
\label{alg:app-active-stream-game}
{\footnotesize
\begin{algorithmic}[1]
\State $(\Gamma_A,\Gamma_B,\mathcal E_{\mathrm{cl}},\sigma)
       \gets\mathcal A(1^\lambda)$
\State $(\mathsf{st}_A,\mathsf{st}_B)\gets
       \Init(1^\lambda,\Gamma_A,\Gamma_B,\mathcal E_{\mathrm{cl}})$
\State $\mathsf{Bad}\gets0$; initialize the integrity monitor below
\For{$d\in\{A\to B,B\to A\}$}
  \State $C^{0,S}_d,C^S_d,C^R_d\gets\epsilon$;
         $\mathsf{sync}_d\gets1$
\EndFor
\State run $\mathcal A^{\mathcal O^b_S,\mathcal O^b_R}(\sigma)$ until finalization
\State disclose $(e_A^\star,e_B^\star)$; obtain $\mathcal A$'s guess $b'$
\State \Return $\mathbf 1[b'=b]$
\end{algorithmic}
}
\end{algorithm}
Legal queries follow the epoch order of Algorithm~\ref{alg:epoch}: one send
per live endpoint, followed by one receive per endpoint (where $\epsilon$
models no delivery), after which the challenger applies the combined close
bits absorbingly. Arbitrary fragmentation, coalescing, delay, insertion, and
deletion are expressed by the receive strings across epochs. With
$p_X(t)=\Gamma_X(t)$ for an endpoint live at the start of $t$, and $p_X(t)=0$
otherwise, the send oracle is Algorithm~\ref{alg:app-active-stream-send}.

\begin{algorithm}[t]
\caption{Active datastream send oracle $\mathcal O^b_S(X,t,\mu)$}
\label{alg:app-active-stream-send}
{\footnotesize
\begin{algorithmic}[1]
\State $Y\gets\overline X$
\State $(c^0,\mathsf{cl}_X^S)\gets\Send_X(t,\mathsf{st}_X,\mu)$
\State \textbf{if} $b=0$ \textbf{then} $c\gets c^0$
       \textbf{ else} $c\gets\Rand(p_X(t))$
\State $C^{0,S}_{X\to Y}\gets C^{0,S}_{X\to Y}\Vert c^0$
\State $C^S_{X\to Y}\gets C^S_{X\to Y}\Vert c$
\State \Return $(c,\mathsf{cl}_X^S)$
\end{algorithmic}
}
\end{algorithm}
Thus $|C^{0,S}_d|=|C^S_d|$ at every point. If a synchronized fragment $u$
occupies positions $i,\ldots,j$ in $C^S_d$, define
$\mathsf{Map}_d(u)=C^{0,S}_d[i..j]$. Let
$\widehat{\mathsf{cl}}_X(t)$ be the close bit computed from the shadow state at
the point of return, after any preceding mapped-honest update. The receive oracle below sets
$d=Y\to X$, and $u$ is the longest prefix of $c$ for which
$C^R_d\Vert u\preceq C^S_d$; write $c=u\Vert v$. The receive oracle is
Algorithm~\ref{alg:app-active-stream-recv}.

\begin{algorithm}[t]
\caption{Active datastream receive oracle $\mathcal O^b_R(X,t,c)$}
\label{alg:app-active-stream-recv}
{\footnotesize
\begin{algorithmic}[1]
\State $Y\gets\overline X$; $d\gets Y\to X$
\If{$\neg\mathsf{sync}_d$}
  \State $C^R_d\gets C^R_d\Vert c$
  \State \textbf{if} $b=1$ \textbf{then}
         \Return $([\,],\widehat{\mathsf{cl}}_X(t))$
  \State $(\mathcal L,\mathsf{cl})\gets\Recv_X(t,\mathsf{st}_X,c)$
  \State \textbf{if} $\mathcal L\ne[\,]$ \textbf{then}
         $\mathsf{Bad}\gets1$
  \State \Return $(\mathcal L,\mathsf{cl})$
\EndIf
\State compute $u,v$ as above
\State \textbf{if} $u\ne\epsilon$ \textbf{then}
       $u^0\gets\mathsf{Map}_d(u)$; $\Recv_X(t,\mathsf{st}_X,u^0)$
\State $C^R_d\gets C^R_d\Vert c$
\State \textbf{if} $v=\epsilon$ \textbf{then}
       \Return $([\,],\widehat{\mathsf{cl}}_X(t))$
\State $\mathsf{sync}_d\gets0$
\State \textbf{if} $b=1$ \textbf{then}
       \Return $([\,],\widehat{\mathsf{cl}}_X(t))$
\State $(\mathcal L,\mathsf{cl})\gets\Recv_X(t,\mathsf{st}_X,v)$
\State \textbf{if} $\mathcal L\ne[\,]$ \textbf{then}
       $\mathsf{Bad}\gets1$
\State \Return $(\mathcal L,\mathsf{cl})$
\end{algorithmic}
}
\end{algorithm}
All shadow sender and receiver calls in this experiment feed their
instrumented generation and acceptance traces to the integrity monitor in
the next subsection. Accordingly, synchronized plaintext is suppressed in both worlds. At the
first deviation only its common prefix is mapped and processed in the random
world; that direction then remains permanently out of sync. The real-world
branch continues to process deviating bytes, while the random branch returns
only the shadow close bit. The challenger also runs the semantic integrity
monitor below; hence forged control state and early close set
$\mathsf{Bad}$ even when $\mathcal L=[\,]$. Active security requires both the
distinguishing advantage and $\Pr[\mathsf{Bad}=1\mid b=0]$ to be negligible.

\subsection{Wrapper-state integrity game}

The integrity game makes the semantic monitor explicit. It instruments the
honest implementation only to expose generated and accepted protected frames;
this instrumentation is not part of the protocol interface. For every
direction $d$ and authenticated semantic sequence $s$, the send oracle stores
the complete honestly generated frame in $\mathsf{Log}_d[s]$.
Algorithm~\ref{alg:app-bd-int} gives the game.

\begin{algorithm}[t]
\caption{Wrapper-state integrity experiment
$\mathsf{Exp}^{\mathrm{BD\mbox{-}INT}}_{\Pi,\mathcal A}$}
\label{alg:app-bd-int}
{\footnotesize
\begin{algorithmic}[1]
\State $(\Gamma_A,\Gamma_B,\mathcal E_{\mathrm{cl}},\sigma)
       \gets\mathcal A(1^\lambda)$
\State $(\mathsf{st}_A,\mathsf{st}_B)\gets
       \Init(1^\lambda,\Gamma_A,\Gamma_B,\mathcal E_{\mathrm{cl}})$
\State $\mathsf{Bad}\gets0$
\State \textbf{for all} $d,s$: $\mathsf{Log}_d[s]\gets\bot$
\State \textbf{for all} $d$: $\mathsf{Used}_d\gets\emptyset$
\State $\mathcal A^{\mathcal O_S,\mathcal O_R}(\sigma)$
\State \Return $\mathsf{Bad}$
\end{algorithmic}
}
\end{algorithm}
On $\mathcal O_{\mathsf{Send}}(X,t,\mu)$, the challenger runs the honest
sender and records every protected frame $F$ created by the call as
$\mathsf{Log}_{X\to\overline X}[s]\gets F$ at its authenticated sequence
$s$.
If the call reports a first visible close when the corresponding
predicate in Eq.~\eqref{eq:streamclose} or Eq.~\eqref{eq:dglinger} is false,
it sets $\mathsf{Bad}\gets1$, then returns the ordinary ciphertext and close
bit.

On $\mathcal O_{\mathsf{Recv}}(X,t,c)$, let $d=\overline X\to X$. The
challenger runs the honest receiver and examines its instrumented acceptance
trace in order. For each accepted pair $(s,F)$, it applies
Algorithm~\ref{alg:app-integrity-monitor}.

\begin{algorithm}[t]
\caption{Per-frame wrapper-integrity monitor for accepted $(s,F)$}
\label{alg:app-integrity-monitor}
{\footnotesize
\begin{algorithmic}[1]
\Require direction $d$, accepted $(s,F)$, semantic state before and after $F$
\State \textbf{if} $\mathsf{Log}_d[s]\ne F$ \textbf{then}
       $\mathsf{Bad}\gets1$
\State \textbf{if} $s\in\mathsf{Used}_d$ \textbf{then}
       $\mathsf{Bad}\gets1$ \textbf{ else}
       $\mathsf{Used}_d\gets\mathsf{Used}_d\cup\{s\}$
\State \textbf{if} $F$ delivers DATA after accepted peer FIN
       \textbf{ then} $\mathsf{Bad}\gets1$
\State \textbf{if} $F$ changes FIN/ACK state without its logged flag
       \textbf{ then} $\mathsf{Bad}\gets1$
\end{algorithmic}
}
\end{algorithm}
After the trace, if the returned delivery list is not exactly the ordered
payloads of the accepted logged DATA frames, the oracle sets
$\mathsf{Bad}\gets1$. Finally, if this call reports the endpoint's first visible close while its
deterministic close predicate is false or $t\notin\mathcal E_{\mathrm{cl}}$,
the oracle sets $\mathsf{Bad}\gets1$, and returns the ordinary delivery and
close bit. Thus the monotone flag covers frame, DATA, FIN/ACK, replay,
post-FIN-DATA, and early-close violations. In the datastream active game it
is ORed with the non-null-after-deviation event shown above.

For the datastream construction, let $\mathsf{InnerBad}$ be an unauthentic
inner plaintext output and $\mathsf{AEADForge}$ a fresh accepted wrapper
ciphertext. If neither occurs, every accepted $(s,F)$ equals the unique
logged frame at that direction and sequence, after which the deterministic
state machine makes all remaining checks hold. Hence
\begin{equation}
 \Pr[\mathsf{Bad}]
 \le \Pr[\mathsf{InnerBad}]+\Pr[\mathsf{AEADForge}],
\end{equation}
Writing $\Adv^{\mathrm{BD\mbox{-}INT}}_{\Pi,\mathcal A}$ for the resulting
bad-event probability, $\Adv^{\mathrm{FEP\mbox{-}CCFA}}$ for the inner
datastream FEP's active-security advantage, and $\mathcal B_1,\mathcal B_2$
for the reductions derived from $\mathcal A$, guessing the attacked direction
with real/random normalization gives the concrete (non-tight) bound
\begin{align}
 \Adv^{\mathrm{BD\mbox{-}INT}}_{\Pi,\mathcal A}
 &\le 4\Adv^{\mathrm{FEP\mbox{-}CCFA}}_{\Pi_{\mathrm{UD}},\mathcal B_1}\nonumber\\
 &\quad+2\Adv^{\mathrm{INT\mbox{-}CTXT}}_{\mathrm{AEAD},\mathcal B_2}
 +\mathsf{negl}(\lambda).
\end{align}
For the datagram construction, an unlogged fresh accepted semantic frame is
instead an atomic FEP-CCA forgery; $R_n$ and $R_s$ make a logged replay null.
Its bad probability is therefore bounded by the sum of the two directional
FEP-CCA forgery probabilities and the transmitted-nonce collision
probability.

\subsection{Active datagram challenge}

Datagrams use atomic rather than prefix synchronization. The send oracle is
the active send oracle above, except that for each direction it records an
initially empty multimap $T_d$ from each shown $c$ to its real shadow
counterpart $c^0$.
A lookup may select any matching counterpart. Multiple
matches among short chaff are harmless because every counterpart decodes to
null; a collision involving an authenticated datagram has negligible
probability. Algorithm~\ref{alg:app-active-dg-recv} defines the receive
oracle.

\begin{algorithm}[t]
\caption{Active datagram receive oracle
$\mathcal O^{b,\mathrm{DG}}_R(X,t,c)$}
\label{alg:app-active-dg-recv}
{\footnotesize
\begin{algorithmic}[1]
\State $Y\gets\overline X$; $d\gets Y\to X$
\If{$\exists c^0\in T_d[c]$}
  \State select any matching counterpart $c^0\in T_d[c]$
  \State $(\mathsf{st}_X,\cdot,\mathsf{cl})\gets\Recv_X(t,\mathsf{st}_X,c^0)$
  \State \Return $(\bot,\mathsf{cl})$
\EndIf
\State \textbf{if} $b=1$ \textbf{then}
       \Return $(\bot,\widehat{\mathsf{cl}}_X(t))$
\State $(\mathsf{st}_X,x,\mathsf{cl})\gets
       \Recv_X(t,\mathsf{st}_X,c)$
\State \Return $(x,\mathsf{cl})$
\end{algorithmic}
}
\end{algorithm}
Thus an honestly shown datagram advances the shadow receiver but its semantic
output is suppressed in both worlds. A new adversarial datagram is opened
only in the real world; any accepted DATA/FIN/ACK or early close is recorded
by the same integrity monitor. Independent queries allow recovery after a
modified or lost datagram, unlike permanent datastream desynchronization.
Short chaff and authenticated-null datagrams produce $\bot$ in both worlds,
and repeated honest outputs cannot redeliver DATA because of $R_n$ and $R_s$.
This is exactly the bidirectional product of the FEP-CCA experiment
of~\cite{fenske2024bytes}, with independent direction keys and the two
close-bucket leakage values.

\section{Construction Details and Invariants}
\label{app:construction}

\subsection{Datastream sender and FIN drain}

The sender maintains an inner plaintext queue $\mathsf{buf}$, generated
ciphertext queue $\mathsf{obuf}$, inner record counter $r$, wrapper counter
$s$, total released bytes $\mathsf{out}$, and two optional FIN counters.
$\mathsf{finrem}$ is the inner plaintext distance through the end of the FIN
object; when record generation consumes that position the sender sets
$\mathsf{finlim}=\mathsf{out}+|\mathsf{obuf}|+|d_{\mathrm{UD}}|$, and FIN is drained
exactly when $\mathsf{out}\ge\mathsf{finlim}$. The two-counter distinction
prevents FIN from being considered public merely because it was queued or
encrypted; the entire containing inner ciphertext must have left the
endpoint API.

The receiver first runs the inner FEP, appends returned plaintext to its
wrapper buffer, and repeatedly removes a complete declared object before
authenticating it. On wrapper failure it sets private bad state, leaves the
expected wrapper sequence unchanged, and stops the call. A valid DATA after
peer FIN, a duplicate FIN, a non-empty FIN, or an unknown tag is a semantic
failure and produces neither DATA nor close.

\subsection{Datagram close safety}

\begin{lemma}[Datagram close safety]
\label{lem:dgclose}
No network adversary can cause visible close before an honest endpoint has
requested FIN, authenticated peer FIN and peer ACK, emitted ACK, completed
its public linger window, and reached
a public close epoch, except with the atomic channel's forgery probability.
\end{lemma}

\begin{proof}
The readiness predicate of Eq.~\eqref{eq:dgclose} is a conjunction of four
local monotone bits. $\mathsf{ownFin}$ is set only by the application
interface. $\mathsf{peerFin}$ and $\mathsf{ownFinAck}$ are set only after an
authenticated semantic frame under the incoming direction key; a novel such
frame is a FEP-CCA/INT-CTXT forgery unless honestly emitted.
$\mathsf{peerAckSent}$ is set only when the endpoint itself releases an
ACK-bearing datagram. The close test additionally checks the public epoch.
The remaining conjunct of Eq.~\eqref{eq:dglinger} counts permitted epochs
since the locally recorded readiness epoch, a value the adversary cannot
manipulate except by delaying readiness itself, which only defers close.
Dropping or replaying traffic cannot set a missing bit.
\end{proof}

Close \emph{liveness} is conditional: readiness requires accepted evidence,
and a visibly closed endpoint stops transmitting. The linger phase retains
FIN/ACK transmission through $L$ further public close buckets when eligible
epochs and sequence space remain. A peer can be stranded if no such frame is
accepted because of loss or modification, a schedule below 40 bytes, or
semantic-sequence exhaustion (\S\ref{sec:eval-loss}). This is the datagram analogue of TCP's
final-ACK problem and its \textsc{time-wait} remedy;
\S\ref{sec:discussion} discusses the residual latency/liveness trade.

\section{Detailed Evaluation Matrices}
\label{app:evaluation}

The correctness sweep used application lengths
$0,1,1023,1024,1025,4011,4012,4013,16507$ under six delivery policies
(direct, one-byte fragments, boundary fragments, seeded random fragments,
one-epoch delay, coalescing after hold). The exact-length sweep covered
selected idle/data boundary lengths between 1 and 8{,}192 bytes, asymmetric
schedules, and a varying asymmetric sequence. Semantic probes checked
wrapper tag failure, higher sequence after failure, unknown authenticated
tags, post-FIN DATA, DATA/FIN tag substitution, inner fail-stop, and
incomplete-length resource quarantine.

The TCP Section-9-style suite used two directions, target emissions 1, 2, 3,
5, and 10, and beginning/middle/end bit positions: 30 paired trials.
Separate insertion, deletion, and truncation runs kept API and
proxy-collected bytes distinct from post-mutation delivered bytes. All
reverse directions completed.

Datagram exact-length tests cover every integer $0\le p\le1500$ plus the
65,507-byte format maximum---dense coverage around realistic MTUs; the
construction's arithmetic enforces the remaining values. Active trials
rotate bit flip, truncation, extension, and uniform replacement; each
checks no target delivery, complete reverse delivery, later target
recovery, and replay rejection.

The parameter sweeps (\S\ref{sec:eval-perf}--\ref{sec:eval-timed}) cover ten
load points, symmetric schedules $\Gamma\in\{200,\dots,12800\}$ at 50\%
utilization, 19 deterministic FIN/ACK drop patterns each under linger depths
$L\in\{0,\dots,3\}$, and 5--100\,ms timed
batches (30 trials $\times$ 10 epochs each), all with oracle checks.

\end{document}